\documentclass[sigconf, nonacm]{acmart}

\newcommand\vldbdoi{XX.XX/XXX.XX}
\newcommand\vldbpages{XXX-XXX}
\newcommand\vldbvolume{14} 
\newcommand\vldbissue{1}
\newcommand\vldbyear{2020}
\newcommand\vldbauthors{\authors}
\newcommand\vldbtitle{\shorttitle} 
\newcommand\vldbpagestyle{plain}

\usepackage{rotating}
  \usepackage{lscape} 
    \usepackage{amsthm}
    \usepackage{amsmath}
    \usepackage{amsfonts}
    \usepackage{stmaryrd}  

    \usepackage{bm}

    \usepackage{url}
    \usepackage{hyperref}
  \usepackage{xargs}

    \newtheorem{proprep}{Proposition}
    \newtheorem{proptyrep}{Property}
    \newtheorem{definrep}{Definition}

    \theoremstyle{definition}

\newcommand{\mudag}{RLQDAG}

\newcommand{\muenum}{\texttt{MuEnum}}

\newcommand{\eqnode}[1]{\hm{\bigr[} #1 \hm{\bigr]}}

\newcommand{\antiprojection}{\widetilde{\pi} }

\newcommand{\strict}{\texttt{rigid}}
\newcommand{\destab}{\texttt{destab} }
\newcommand{\dee}{$\texttt{destab}$}
\newcommand{\deriv}[2]{\texttt{d}(#1, #2)}

\newcommand{\CASE}[1]{\STATE \textbf{case} #1\textbf{:} \begin{ALC@g}}
  \newcommand{\ENDCASE}[1]{\textbf{end case}\end{ALC@g}}

  \newcommand{\evaluation}[1]{\texttt{eval}({#1})}
  \newcommand{\pf}[1]{\text{pf}(#1)}
  \newcommand{\pa}[1]{\text{pa}(#1)}
  \newcommand{\pj}[1]{\text{pj}(#1)}
  \newcommand{\mf}[1]{\text{mf}(#1)}
  \newcommand{\pp}[1]{\text{pp}(#1)}
  \newcommand{\expand}[1]{\texttt{expand}(#1)}
  \newcommand{\expandR}[2]{\texttt{expand}_{#1}(#2)}
  
  \newcommand{\applyAll}[1]{\texttt{applyAll}(#1)}
  \newcommand{\allCodd}[1]{\texttt{allCodd}(#1)}
  \newcommand{\pfj}[1]{\text{pfj}(#1)}
  \newcommand{\jassoc}[1]{\text{jassoc}(#1)}
  \newcommand{\dju}[1]{\text{dju}(#1)}
  \newcommand{\paj}[1]{\text{paj}(#1)}
  
  \newcommand{\filtColumns}[1]{\text{\emph{filt(#1)}}}
  \newcommand{\type}[1]{\text{\emph{type}}{(#1)}}
  \newcommand{\recursiveNode}[1]{\eqnode{#1}_{\mathfrak{D}}^{\mathfrak{R}}}
  \newcommand{\recursiveNodeAnot}[3]{\eqnode{#1}_{\mathfrak{#2}}^{\mathfrak{#3}}}

  \newcommand{\sd}[1]{S_d\llbracket #1 \rrbracket}
  \newcommand{\salpha}[2]{S_{\alpha}\llbracket #1 \rrbracket_{#2}}
  \newcommand{\sgamma}[2]{S_{\gamma}\llbracket #1 \rrbracket_{#2}}

\newcommand\sbullet[1][.5]{\mathbin{\vcenter{\hbox{\scalebox{#1}{$\bullet$}}}}}

\newcommand{\subst}[3]{ {#1}_{\{{#2}/{#3}\}}}

\let\phi\varphi
\let\othcirc\circ
\def\circ{\mathbin\othcirc}
\newcommand*{\filt}[2][\mathfrak f]{{\sigma_{#1}\left(#2\right)}}

\newcommand*{\rename}[3]{\rho_{\mathsf{#1}}^{\mathsf{#2}}\left(#3\right)}

\newcommand*{\fixpt}[2][X]{\mu {\left(#1=#2\right)}}

\newcommand{\NJoin}{\bowtie}

\newcommand{\filtaf}[2]{\sigma_{#1}({#2})}
\newcommand{\renameaf}[3]{\rho_{#1}^{#2}({#3})}

\newcommand{\syntaxdef}{\mathrel{::=}}

\newcommand{\syntaxtable}[1]{
  \def\entry##1[##2]##3[##4]{
    {##1} & \syntaxdef & \hspace{3cm} & \!\!\!\! \mbox{##2}
  \\    &     & {##3} & \mbox{##4} }
  \def\singleentry##1[]##2[##3]{
  {##1} & \syntaxdef & {##2} & \!\!\!\! \mbox{##3} }
  \def\oris##1[##2]{
    \\    & |   & {##1} & \mbox{##2} }
  \def\orisopt##1[##2]{
    \\ \left(   & |   & {##1} & \mbox{##2} \right) }
  \def\linekot##1[##2]{
      \\  & -- & ----------- & -----
   }
  \def\linestart[##1]##2{ %not yet used properly 
   \\ \mbox{##1} & \hspace{3cm} & \!\!\!\! & {##2} 
  }
  \begin{array}{rcll}
  #1
  \end{array}
  }
\newcommand{\smallsyntax}[1]{\[\syntaxtable{#1}\]}
\newcommand{\consistent}[4]{\texttt{cons}(#1,#2)_{#3}^{#4}}

\newcommand{\replaceparam}{\text{\emph{rep}}}

\newcommand{\pushedf}{\text{\emph{\textcolor{blue}{pushed}}}}
\newcommand{\unpushedf}{\text{\emph{\textcolor{black}{unpushed}}}}
\newcommand{\unpushedftwo}{\text{\emph{\textcolor{black}{unpushed}}}_2}
\newcommand{\expandfalphatwo}{\text{\emph{\textcolor{black}{expand}}}_{\alpha_2}}

\newcommand{\pushedp}{\text{\emph{\textcolor{blue}{pushed}}}}
\newcommand{\unpushedp}{\text{\emph{\textcolor{black}{unpushed}}}}
\newcommand{\unpushedptwo}{\text{\emph{\textcolor{black}{unpushed}}}_2}
\newcommand{\expandpalphatwo}{\text{\emph{\textcolor{black}{expand}}}_{\alpha_2}}

\newcommand{\unaryop}[1]{\text{op}_\text{u}(#1)}
\newcommand{\binaryop}[2]{\text{op}_\text{b}(#1, #2)}
\newcommand{\alphasub}{\alpha_\text{sub}}
\newcommand{\wellformed}{(WF)}
\newcommandx{\fixptdag}[5][1=X, 4=\mathfrak{D}, 5=\mathfrak{R}]{\ensuremath{\mu (#1.#2\cup\eqnode{#3}_{#4}^{#5})}}
\newcommand{\consppty}[1]{\texttt{consistent}(#1)}
\newcommandx{\anneqnode}[3][2=\mathfrak{D}, 3=\mathfrak{R}]{\ensuremath{\eqnode{#1}_{#2}^{#3}}}
\newcommand{\pushedfb}{\text{\emph{pushed}}}
\newcommand{\condpj}{\text{\emph{cond}}_\text{pj}}
\newcommand{\wfexpandp}{(R1)}
\newcommand{\pconsexpand}{(R2)}

\begin{document}
\title{Efficient Enumeration of Recursive Plans in Transformation-based Query Optimizers}

% \thanks{This research has been supported by the ANR project GraphRec (ANR-23-CE23-0010).}

\author{Amela Fejza}
\affiliation{%
  \institution{Tyrex team,
  Univ. Grenoble Alpes, CNRS, Inria,
  Grenoble INP, LIG}
  % \streetaddress{P.O. Box 1212}
  \city{Grenoble}
  \postcode{38000}
  \country{France} 
}
\email{amela.fejza@inria.fr}

\author{Pierre Genev\`es}
\affiliation{%
  \institution{Tyrex team,
  Univ. Grenoble Alpes, CNRS, Inria,
  Grenoble INP, LIG}
  % \streetaddress{P.O. Box 1212}
  \city{Grenoble}
  \postcode{38000}
  \country{France}
}
\email{pierre.geneves@inria.fr}

\author{Nabil Laya\"ida} 
\affiliation{%
  \institution{Tyrex team,
  Univ. Grenoble Alpes, CNRS, Inria,
  Grenoble INP, LIG}
  % \streetaddress{P.O. Box 1212}
  \city{Grenoble}
  \postcode{38000}
  \country{France}
}
\email{nabil.layaida@inria.fr}

\begin{abstract}
  Query optimizers built on the transformation-based Volcano/Cascades framework are used in many database systems. Transformations proposed earlier on the logical query dag (LQDAG) data structure, which is key in such a framework, are restricted to recursion-free queries. We propose the recursive logical query dag (\mudag{}) which extends the LQDAG with the ability to capture and transform recursive queries, leveraging recent developments in recursive relational algebra. Specifically, this extension includes: (i) the ability of capturing and transforming sets of recursive relational terms thanks to (ii) annotated equivalence nodes used for guiding  transformations that are more complex in the presence of recursion; and (iii) \mudag{} rewrite rules that transform sets of subterms in a grouped manner, instead of transforming individual terms in a sequential manner; and that (iv) incrementally update the necessary annotations. Core concepts of the \mudag{} are formalized using a syntax and formal semantics with a particular focus on subterm sharing and recursion.
  The result is a clean generalization of the LQDAG transformation-based approach, enabling more efficient explorations of plan spaces for recursive queries. An implementation of the proposed approach shows significant performance gains compared to the state-of-the-art.
\end{abstract}

\maketitle


\pagestyle{\vldbpagestyle}
\begingroup\small\noindent\raggedright\textbf{PVLDB Reference Format:}\\
\vldbauthors. \vldbtitle. PVLDB, \vldbvolume(\vldbissue): \vldbpages, \vldbyear.\\
\href{https://doi.org/\vldbdoi}{doi:\vldbdoi}
\endgroup
\begingroup
\renewcommand\thefootnote{}\footnote{\noindent
This work is licensed under the Creative Commons BY-NC-ND 4.0 International License. Visit \url{https://creativecommons.org/licenses/by-nc-nd/4.0/} to view a copy of this license. For any use beyond those covered by this license, obtain permission by emailing \href{mailto:info@vldb.org}{info@vldb.org}. Copyright is held by the owner/author(s). Publication rights licensed to the VLDB Endowment. \\
\raggedright Proceedings of the VLDB Endowment, Vol. \vldbvolume, No. \vldbissue\ %
ISSN 2150-8097. \\
\href{https://doi.org/\vldbdoi}{doi:\vldbdoi} \\
}\addtocounter{footnote}{-1}\endgroup

%\ifdefempty{\vldbavailabilityurl}{}{
%\vspace{.3cm}
%\begingroup\small\noindent\raggedright\textbf{PVLDB Artifact Availability:}\\
%The source code, data, and/or other artifacts have been made available at \url%{\vldbavailabilityurl}.
%\endgroup 
%}
% %%% VLDB block end %%%

\section{Introduction}
\label{section:introduction}

Recursive queries enable powerful information extraction, especially from linked data structures such as trees and graphs.  
However, important data management system components, such as the widely used Volcano framework \cite{volcano-1993}, were designed for recursion-free queries.
A typical transformation-based query optimizer operates by (i) translating a query into a relational algebraic term, (ii) applying algebraic transformations in order to search for equivalent yet more efficient evaluation plans, during a so-called \emph{plan enumeration phase}, (iii) executing the query by running one of the explored plans. %
Works on extending relational algebra (RA) with recursion \cite{abiteboul-1995,bancilhon-ramakrishnan-recursion} resulted in powerful recursive relational algebras \cite{LFP-RA,alpha-RA,geneves-sigmod20}, capable of capturing queries with transitive closures \cite{alpha-RA} and even more general forms of recursion \cite{LFP-RA,geneves-sigmod20}. This line of works recently led to $\mu$-RA \cite{geneves-sigmod20} which provides a rich set of rewrite rules for recursive terms enabling efficient evaluation plans not available with earlier approaches.
% The speed of plan space exploration represents one of the most critical aspect in query optimization.
The enumeration phase is crucial as it may produce terms which are drastically more efficient. 
It has been well-studied for recursion-free queries. 
Allocating a time budget for this enumeration phase is common, as it is notoriously known that exhaustive plan space explorations may not be feasible in practice for certain queries.
The faster we generate the space of equivalent plans, the more likely we will be able to find more efficient plans.
With recursion, plan spaces are often significantly larger than in the non-recursive setting, due to new interplays between recursive and non-recursive operators. The efficiency of recursive plan enumeration becomes critical. 
Plan enumeration speed directly determines whether query evaluation plans theoretically enabled by e.g. $\mu$-RA \cite{geneves-sigmod20} are within range of a practical query optimizer or not.
This motivates the search for efficient recursive plan enumeration methods, as they unlock potential for big performance gains, and are possibly decisive for the feasibility of certain queries.

\paragraph*{Contributions}
We present the \underline{R}ecursive \underline{L}ogical \underline{Q}uery \underline{Dag}, \mudag{}, which extends Volcano's LQDAG \cite{volcano-1993} and the $\mu$-RA framework \cite{geneves-sigmod20} for the purpose of efficiently enumerating recursive query plans. %
Contributions include (i) the first extension of the LQDAG with the support of recursive terms; (ii) a formalization of important \mudag{} concepts in terms of formal syntax and semantics, with a particular focus on the sharing of common subterms in the presence of recursion; (iii) \mudag{} transformations with incremental annotation updates. These transformations generalize rewrite rules from individual recursive terms (such as those of \cite{geneves-sigmod20}) to \emph{grouped transformations of compactly represented sets of recursive terms}. 
This enables much more efficient explorations of recursive plan spaces, which in turn makes available in practice very efficient evaluation plans unmatched by previous techniques. %
For this to be possible, the \mudag{} relies on a concept of annotated equivalence nodes with incremental updates, used for guiding transformations of recursive subterms. Contributions also include (iv) a complete implementation of the proposed approach and experimental assessments using third-party queries on synthetic and real datasets.

\section{Background and Related Work}
\label{section-related_work}

\newcommand{\cleanpara}[1]{\textbf{{\underline{#1}}}}

\label{sec:datalog}
Recursion is considered in only a small fraction of the numerous works on query optimization.
Three main lines of work with recursive queries can be identified.

\cleanpara{The Datalog line of works} \cite{datalog-1, datalog-2,datalog-3,datalog-4,datalog-5,datalog-6,datalog-7,datalog-8} developed methods for optimizing recursive queries formulated in Datalog: magic-sets ~\cite{datalog-magic-sets,datalog-magic-sets-2,datalog-opt}, demand transformations~\cite{demand-driven-datalog-2011}, automated reversals \cite{right-left-linear-datalog}, and the FGH rule \cite{fgh-rule-datalog-2022}. 
     Although the syntax of Datalog greatly differs from RA, the effects of magic-sets \cite{datalog-magic-sets,datalog-magic-sets-2,datalog-opt} and of demand transformations~\cite{demand-driven-datalog-2011} are comparable to pushing certain kinds of selections and projections. These techniques are very sensitive to whether the Datalog program is written in a left-linear or right-linear manner, but one can use the automated reversal technique \cite{right-left-linear-datalog} to fully exploit them. 
    The framework proposed in \cite{fgh-rule-datalog-2022} gathers magic-sets, semi-naive evaluation and proposes a new FGH rule for optimizing recursive Datalog programs with aggregations. 

    Datalog engines do not explore plan spaces but use heuristics to find a good plan to evaluate queries.
    However, currently, no matter which combination of existing Datalog optimizations a Datalog engine implements, it will not be able to find plans where recursions have been merged automatically similar to those found by the $\mu$-RA approach \cite{geneves-sigmod20}. This is because, currently, in a Datalog program corresponding to the optimized translation of A$^+$/B$^+$ at least one of the two transitive closures A$^+$ or B$^+$ will be fully materialized (even if there is no solution to A$^+$/B$^+$). On real datasets, this can make Datalog query evaluation an order of magnitude slower than query evaluation with RA-based systems, as noticed in \cite{geneves-sigmod20}.

  \cleanpara{The line of works based on relational algebra} \cite{codd-ra,LFP-RA,lfp-ra-extension,alpha-RA,geneves-sigmod20} attempts to extend relational algebra with operators to capture forms of recursion. For instance, $\alpha$-RA \cite{alpha-RA} extends RA with an operator to capture transitive closures. LFP-RA \cite{LFP-RA} proposes a more general least fixpoint operator, and an extension of this work gave effective criteria for optimization in its presence~\cite{lfp-ra-extension}. Recently, $\mu$-RA~\cite{geneves-sigmod20} proposed to extend RA with a $\mu$ operator which is also a fixpoint with an appropriate set of restrictions. This enables $\mu$-RA to combine all earlier RA-based recursion optimization rules in the same framework~\cite{geneves-sigmod20}, while adding new rewrite rules for recursive terms, in particular for merging recursions. This makes $\mu$-RA the most advanced system for RA-based recursive query optimization as it can generate plans not reachable by other approaches. The approach that we propose extends $\mu$-RA with a new method for enumerating recursive plans much more efficiently, thanks to a generalization of the rules presented in $\mu$-RA so that the generalized rules (presented in Section~\ref{section: mudag-adapted-rewrite-rules}) apply directly on a factorized representation of the recursive plan space. Among the benefits, this enables (1) applying transformations on sets of algebraic terms at once (instead of successively on individual terms), and (2) exploiting the sharing of common subterms to avoid redundant computations.

  \cleanpara{Ad-hoc optimizations for regular queries.}
      Both Datalog and RA-based approaches can capture queries with expressive forms of recursion, going beyond regularity. Several works have focused on optimizing queries in which recursion is restricted to regular patterns, such as regular path queries (RPQs). 
    Automata based approaches have been developed to answer RPQ queries~\cite{rpq-automata,dataguides-automata,rpl-automata}. A hybrid approach that combines finite state machines and $\alpha$-RA is presented in \cite{waveguide-2015} and extended in~\cite{tasweet-2017}. All these works are limited to RPQs and their unions or conjunctions. In comparison, the work we present in this paper supports more expressive forms of recursion, that may include non-regular patterns (see e.g. the experimental section~\ref{section:experimental_results} with queries of the form $A^nB^n$ for instance).

%\end{itemize}

\cleanpara{Limits of the RA-based approach.}
The RA line of work offers several advantages including high expressivity and rich plan spaces. In addition, it can be seen as an interesting approach for extending the (non-recursive) RA approach already widely used in RDBMS implementations. One main limit however, is that the whole approach critically depends on the ability to quickly enumerate plans during the plan exploration phase. While plan enumeration has not been studied yet in the presence of recursion, it has been extensively studied for recursion-free queries.

\vspace{-0.2cm}
\subsection{Plan enumeration for non-recursive queries} 
Most of the works on plan enumeration focus on select-project-join (SPJ) queries. 
Techniques proposed for SPJ can be divided into two main groups: bottom-up and top-down approaches. Bottom-up approaches \cite{systemR-1979,r-star-query-pro,starburst-grammar,starburst-sigmod,egg,eq-saturation,spores,tensor-optim} generate the plan space by starting from the leaves (initial relations) and going up in the tree of operators when progressively exploring alternate combinations of operators. In contrast, top-down approaches \cite{exodus-1987,volcano-1993,cascades-1995} start from the root and recursively explore subbranches in search for possible alternatives. %
%
%\cleanpara{Bottom-up approaches.}
%Pioneering and seminal works on the bottom-up approach were proposed in~\cite{systemR-1979} with refinements in execution strategies~\cite{r-star-query-pro} and the Starburst plan generator~\cite{starburst-grammar,starburst-sigmod}. 
%More recently, E-Graphs \cite{egg} and equality saturation techniques \cite{eq-saturation} introduced in the context of SMT solvers have been used in the database community to optimize recursion-free relational and tensor algebra \cite{spores,tensor-optim}. E-graph \cite{egg} basically corresponds to the LQDAG  \cite{volcano-1993} (also known as AND-OR-DAG), as noticed in \cite{spores}. In comparison to LQDAG and E-Graph, that do not support recursive queries, RLQDAG adds the support for recursion. E-Graph is fundamentally bottom-up, whereas RLQDAG relies on top-down search to prune the search space. The top-down aspect of RLQDAG is essential to solve the problem of exploring recursive plans, since transformations of sets of recursive subterms critically rely on a complete view of recursive patterns. 
%
%\cleanpara{Top-down approaches.}
%The top-down approach was introduced with EXODUS~\cite{exodus-1987} and the seminal works on Volcano~\cite{volcano-1993} and Cascades~\cite{cascades-1995} that added memoization to increase efficiency. The purpose of memoization is to re-use information already found in order to avoid redundant calculations. 
An advantage of the top-down approach is that it enables branch and bound pruning~\cite{shapiro-top-down-better}.

%\subsubsection{Join enumeration}

%With SPJ queries the crux of plan enumeration is join enumeration \cite{neumann-2018}. The join enumeration problem consists in finding the best join order for a query with $n$ join operators. It is an NP-hard problem~\cite{join-np-hard}, and has been the focus of extensive research. 

%Ono and Lohman showed that the worst-case complexity of join enumeration is $O(3^n)$ (reached with the clique query structure) \cite{vldb90-ono-lohman-complexity}.  Join enumeration can be made as efficient as the lower bound of $O(3^n)$ as long as the generation of duplicates is avoided \cite{pellenkoft-97}; and this is implementable with a technique based on Volcano memoization \cite{pellenkoft-97}. 

%Algorithmic improvements for the bottom-up approach have been researched in 
%\cite{bennet-96,moerkotte-neumann-06,moerkotte-neumann-08}. Top-down approaches have been shown to be competitive with bottom-up ones in \cite{sigmod07-dehaan-join-optimization}. This triggered further research which resulted in many algorithmic improvements of the top-down approach~\cite{sigmod07-dehaan-join-optimization, fender-moerkotte-2011, fender-moerkotte-2012,fender-moerkotte-13-icde,fender-moerkotte-13-vldb}, culminating with \cite{shanbhag-vldb14} 
%which proposes an algorithm based on~\cite{fender-moerkotte-2012} for enumerating, \`a la Volcano, cross-product free trees of join operators in a complete manner.

All the previous approaches focus mainly on SPJ queries, with some
extensions to support outer joins~\cite{fender-moerkotte-13-vldb,fender-moerkotte-13-icde,outerjoin-extension1,outerjoin-extension2,outerjoin-extension3} and aggregations~\cite{aggregation-extension}. %
Very few works consider other operators, as noticed in \cite{surajit-overview-98} and \cite{shanbhag-vldb14}. %In particular, union has been largely neglected in the litterature, with the notable exception of \cite{DNC-union-pushdown-97} that proposes a union pushdown technique for disjunctive queries. The idea is that union can be placed immediately after the join operators whenever the selectivity of the join operator is very low. 
To the best of our knowledge, plan enumeration for transformation-based query optimizers has not been studied yet in the presence of recursive operators. %
%
%The reason why join enumeration constitutes the crux of the problem of enumerating plans for SPJ queries is because with those queries that are union-free and recursion-free, selections and projections can be abstracted. Selections and projections can be pulled in the tree of operators of an SPJ, reaching a phase where only the remaining join operators are under scrutiny, and once the problem of finding a good join order is solved, selections and projections can be pushed back as close to the sources as possible. In the presence of recursion however, the situation is not comparable anymore, because selections and projections cannot be pulled and pushed back so easily. This makes recursion-free plan enumeration techniques hardly reusable in the presence of recursion. %

Union and recursion greatly extend the expressive power of SPJ queries. However, they not only make plan enumeration significantly more complex, but they also generate significantly larger plan spaces. This is because their addition generates many new possible combinations to be explored due to new interplays, for instance between unions and joins (e.g. distributivity of natural join over union) or between recursions and joins. This worsens the combinatorics of plans to be enumerated and motivates even more the interest of finding efficient techniques for enumerating recursive plans.

\subsection{The logical query dag (LQDAG)} 

The method that we propose extends a key component used in the top-down enumeration approach: the logical query dag (LQDAG). The LQDAG is a directed acyclic graph data structure used to represent and generate the logical plan space. It was introduced in~\cite{volcano-1993} and improved in~\cite{cascades-1995}. It is also well described as the AND-OR-DAG in~\cite{prasan-2000,shanbhag-vldb14} where it is used for detecting and unifying common subexpressions for multi-query optimization \cite{prasan-2000}; and for generating the space of cross-product free join trees \cite{shanbhag-vldb14}.

The LQDAG contains nodes of two different types: \emph{equivalence} nodes and \emph{operation} nodes. Equivalence nodes can only have operation nodes as children and vice versa: operation nodes can only have equivalence nodes as children. The purpose of an equivalence node is to explicitly regroup equivalent subterms. An operation node corresponds to an algebraic operation like: join ($\bowtie$), filter ($\sigma_{\theta}$) etc.  The LQDAG can be seen as a factorized representation of a set of terms. %
Inspired from~\cite{prasan-2000}, Figure~\ref{fig:mudag-expanded-joins} illustrates a sample LQDAG and its expansion obtained after the application of commutativity and associativity rules for the join operator.  
\begin{figure}[h]
    \centering
    \includegraphics[width=0.5\linewidth]{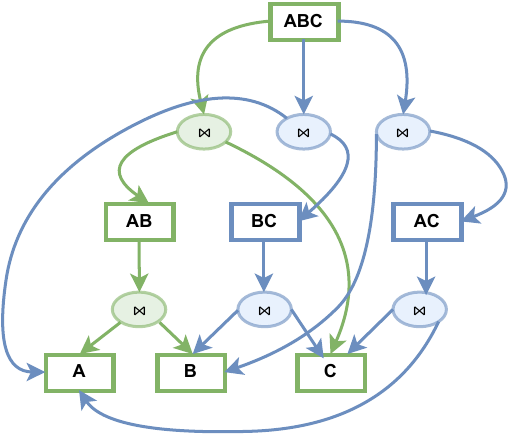}
    \vspace{-0.3cm}
    \caption{Sample LQDAG (in green) with expansion (in blue).}
    \label{fig:mudag-expanded-joins} 
  \end{figure} 

  \vspace{-0.4cm}

\section{The \mudag{}}

%The \mudag{} extends the LQDAG with the ability to capture and transform sets of recursive terms. 
The \mudag{} introduces a novel representation designed to effectively capture and transform sets of recursive terms. This representation enables the delineation of subclasses within recursive terms, based on shared properties. This segmentation facilitates the grouping of similar terms, streamlining the process of collectively transforming them in a single operation. %

\label{sec:mudag-syntax}

\vspace{-0.1cm}\subsection{Syntax}
First of all, we propose a syntax for the \mudag{}. The purpose is to be able to syntactically express a term that denotes a (potentially very large) set of recursive algebraic terms. This makes it possible to develop transformations of sets of terms formally (i.e. with a high level of precision), and express them as rewrite rules that transform one \mudag{} term $d$ into another \mudag{} term $d'$. %
The syntax of \mudag{} terms, given in Fig.~\ref{fig:syntax_muLQDAG}, focuses on formalizing the concepts of equivalence nodes, operation nodes, sharing of common subterms, and recursion. %
\begin{figure}[h]
  \vspace{-0.25cm}
  \begin{footnotesize}\smallsyntax{\entry {\gamma}  [{Pointer to equivalence node}] { \hm{\bigr[} \alpha \hm{\bigr]} } [Equivalence node]
  \oris {\text{ Y }} [Reference] \\
  \\ 
  \entry {\alpha}  [{Equivalence node internals}] { d } [Operation node]
  \oris {d, \alpha} [Operation nodes] 
  \oris {\text{let Y = $\alpha_1$ in $\alpha_2$}} [Reference binder] \\
  \\   
  \entry {d} [{Operation node}] 
  X  [Relation variable]
  \oris \renameaf{a}{b}{{\gamma}} [Rename]
  \oris \filtaf{f}{{\gamma}} [Filter] 
  \oris {\gamma} \bowtie {\gamma'} [Join]
  \oris {\gamma} \triangleright {\gamma'} [Antijoin]
  \oris {\gamma} \cup {\gamma'} [Union]
  \oris \antiprojection_{a}{({\gamma})} [Antiprojection] 
  \oris \mu X. \thinspace \gamma \cup \alpha_{rec} [Fixpoint (recursion)] \\
  \\
    \singleentry {\alpha_{rec}}  [] {\eqnode{\alpha}_{\mathfrak{D}}^{\mathfrak{R}}} [{Annotated equivalence node}]
 } \end{footnotesize}%
 \vspace{-0.4cm}\caption{Syntax of \mudag{} terms.}
    \label{fig:syntax_muLQDAG} \vspace{-0.4cm}
        \end{figure}
 
        An equivalence node is a node that can have several operation nodes $d$ as children, possibly with binders. The binder construct 
        ``$\text{let}~Y~=~\alpha_1~\text{in}~\alpha_2$'' enables the explicit sharing of a common equivalence node $\alpha_1$ within the branches of another equivalence node $\alpha_2$. For that purpose, it assigns a new fresh reference name $Y$ to $\alpha_1$, and allows $Y$ to be used multiple times in $\alpha_2$ as a reference to $\alpha_1$. Hence, the general definition of $\gamma$ is either an equivalence node $\eqnode{\alpha}$ or a reference $Y$ to an existing equivalence node.
        
        Operation nodes are defined by the variable $d$ in the abstract syntax of Fig.~\ref{fig:syntax_muLQDAG}. They include the main algebraic operations of recursive relational algebra. Each operand of an operation node  $d$ points in turn to an equivalence node (through $\gamma$). %      
        The rename operator $\rename{a}{b}{\gamma}$ renames column $a$ into column $b$ in the equivalence node %\footnote{As a slight discrepancy between the theory and its implementation, in the prototype implementation we automatically push, combine and simplify renamings so that they appear only in front of variables $X$. So the core syntax for the rename operator in the implementation is  $\rho^{\{a\mapsto b, ..., w \mapsto z\}}(X)$. The implementation also provides a way to denote  $\rename{a}{b}{\gamma}$ as a syntactic sugar that is automatically translated into the core syntax.}.
        $\gamma$.
        The filter operator $\filtaf{f}{\gamma}$ applies the filtering expression $f$ to the equivalence node $\gamma$. 
        The antiprojection operator $\antiprojection_{a}{({\gamma})}$ removes column $a$ from the equivalence node $\gamma$.
 
For example, with this syntax, the LQDAG of Fig.~\ref{fig:mudag-expanded-joins} is written as 
the term $[ [A \bowtie B] \bowtie C]$ before expansion and is written $[[A \bowtie B] \bowtie C, A \bowtie [B \bowtie C], B \bowtie [A \bowtie C]]$ after expansion, 
where for the sake of readability we omit brackets around relation variables.

        Recursive \mudag{}s can be expressed using fixpoint operation nodes. The principle, inspired from earlier works in recursive relational algebras \cite{LFP-RA,geneves-sigmod20} and generalized here to sets of terms, consists in the introduction of a least-fixpoint binder operation node ($\mu$) that binds a fresh variable $X$ to some expression in which  $X$ can appear, thus explicitly denoting recursion. %
        Our generalization is defined in the syntax of Fig.~\ref{fig:syntax_muLQDAG} and illustrated in Fig.~\ref{fig:lqdag_recursion}. A fixpoint operator node is written $\mu X.~ \gamma ~ \cup ~ \alpha_{rec}$. The operand $\gamma$ is an equivalence node that models the constant part (the base case) of the recursion. $X$ cannot occur within $\gamma$. The equivalence node $\alpha_{rec}$ is the recursive part. An essential consideration is that the $\alpha_{rec}$ branch contains at least one free occurrence of the recursive variable $X$. This characteristic distinguishes the fixpoint operation node from other operation nodes. It will lead to a number of new definitions and formal developments. Intuitively, this is because depending on how the recursive variable is used in that branch, transformation and sharing of \mudag{} subterms may, or may not, be allowed.

For example, on the Yago graph dataset \cite{yagooo}, the query $\mathcal{Q}_{e_1}$:
\vspace{-0.05cm}\begin{center}$\begin{small} ?s,~?t~~ \leftarrow ~~ ?s~~ \texttt{isLocatedIn+} ~~?t\end{small}$\end{center}
\noindent retrieves all pairs ($s,t$) of source and target nodes connected by a path composed of a sequence of edges labeled \begin{small}\texttt{isLocatedIn}\end{small} (transitive closure).
  The following \mudag{} term $\Sigma$ corresponds to $\mathcal{Q}_{e_1}$:
    \begin{center}$\begin{small}\mu X.~ \eqnode{\texttt{isLocIn}} ~ \cup ~ \eqnode{\antiprojection_{m}\eqnode{{\renameaf{t}{m}{\eqnode{\texttt{isLocIn}}} \bowtie \renameaf{s}{m}{\eqnode{X}}}}}_{\mathfrak{D}}^{\mathfrak{R}}\end{small}$
    \end{center}
    It makes recursion explicit using the fixpoint operator. 
    The equivalence node for the constant part contains the relation variable $\begin{small}\texttt{isLocIn}\end{small}$ whose column names are $s$ and $t$. 
    The equivalence node for the recursive part is composed of a join between the recursive variable $X$ with the relation variable $\begin{small}\texttt{isLocIn}\end{small}$. Here, the path traversal is performed from right to left, by introducing a temporary column name ``$m$'' and renaming the columns so that the natural join is performed on the only common column ``$m$'' before ``$m$'' is discarded by the antiprojection.

Fig.~\ref{fig:lqdag_recursion} illustrates the \mudag{} of $\mathcal{Q}_{e_1}$ with two recursive subterms $\Sigma$ and $\Sigma'$. Notice that $\Sigma'$ is semantically equivalent to $\Sigma$ and encodes the left to right direction of traversal using a different column renaming in the recursive part.
        \begin{figure}[h]
         \centering
         \includegraphics[width=1\linewidth]{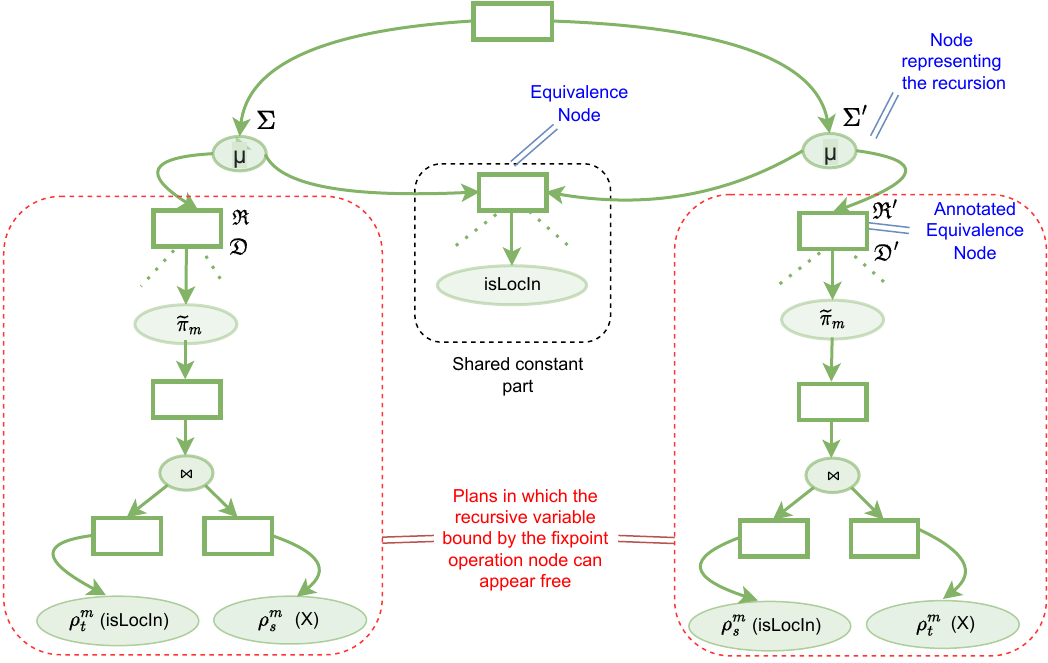}
         \caption{Structure of recursive terms in \mudag{}.}
         \vspace{-0.2cm}
         \label{fig:lqdag_recursion} 
         \vspace{-0.1cm}
       \end{figure} 
 
 The recursive equivalence node of each fixpoint operation node is annotated with  $\mathfrak{D}$ and $\mathfrak{R}$. These annotations will be key for guiding the application of transformations (see Section~\ref{section:criteria-definition}). 
              
       There are two reasons why we distinguish $\alpha_{rec}$ from a general equivalence node $\alpha$ in the abstract syntax. The first reason is that equivalence nodes for recursive parts are equipped with annotations (see Section~\ref{section:annotated_alternative}.). The second reason is that we want to allow a maximum level of sharing while preventing the sharing of subterms with free occurences of a recursive variable. We thus forbid the use of the binding construct to share subterms with free variables.

 We consider the following restrictions over the abstract syntax presented in Fig.~\ref{fig:syntax_muLQDAG}: we consider only positive, linear and non mutually recursive \mudag{} terms:
 (i) positive means that recursive variables only appear in the left-hand operand of an antijoin operator node;
  (ii) linear means that one of the operands in a join or antijoin operator node is constant in the free variable;
  (iii) non-mutually recursive terms means that fixpoint operator nodes are properly nested so that there is only one free variable in any subbranch of an annotated equivalence node (this variable may occur several times).
 These restrictions define a subset of \mudag{} terms that simplifies the theory while supporting expressive queries containing union, conjunction, transitive closure of arbitrary expressions and non-regular patterns such as $A^nB^n$.

\subsection{Semantics of \mudag{} terms}
\label{sec:mudag-semantics}
The interpretation of a \mudag{} term is the set of all recursive relational algebraic terms that it represents. Formally, the semantics of a \mudag{} $[\alpha]$ is a set of $\mu$-RA terms as defined by the functions $\salpha{}{}$ and $\sgamma{}{}$ presented in 
Fig.~\ref{fig:semantics_muLQDAG}, where $E$ denotes a variable environment used to keep track of the variable definitions introduced by binders for the sharing of subterms, and $E \oplus \{ Y \mapsto  \alpha_1 \}$ denotes the environment $E$ in which variable $Y$ is bound to $\alpha_1$. The interpretation of a \mudag{} $[\alpha]$ is $\salpha{\alpha}{\emptyset}$. %
\begin{figure}[h]
  % \vspace{-0.5cm}
  \centering\begin{footnotesize}\renewcommand\arraystretch{1.2} 
$\begin{array}{rcl}
  \sgamma{ [ \alpha ] }{E} & = & \salpha{\alpha}{E}\\
  \sgamma{ Y }{E} & = & \salpha{E(Y)}{E}\vspace{0.3cm}\\
  \salpha{d}{E} & = & \sd{d}  \\
 \salpha{d, \alpha}{E} & = & \sd{ d }{} \cup \salpha{ \alpha }{E}   \\
 \salpha{ \text{let}~ Y = \alpha_1 ~\text{in}~ \alpha_2}{E} & = & \salpha{\alpha_2}{E \oplus \{ Y \mapsto  \alpha_1 \}}~ \vspace{0.3cm}\\
  \sd{X}  & = & X \\
   \sd{\sigma_f(\gamma)} &  = & \{~ \sigma_f(t)  ~|~  t  \in  \sgamma{\gamma}{E}\}  \\ 
 \sd{\gamma_1 \bowtie \gamma_2} & = & \{ ~ t \bowtie t' ~|~ t \in  \sgamma{\gamma_1}{E} ~\land~ t' \in  \sgamma{\gamma_2}{E}   \} \\  \sd{\gamma_1 \triangleright \gamma_2}{} & = & \{~ t  \triangleright  t' ~|~  t \in  \sgamma{\gamma_1}{E} ~\land~ t' \in  \sgamma{\gamma_2}{E} \} \\ 
 \sd{\gamma_1 \cup \gamma_2}{} & = & \{ ~ t \cup t'    ~|~    t \in  \sgamma{\gamma_1}{E} ~\land~ t' \in  \sgamma{\gamma_2}{E}  \} \\
 \sd{\renameaf{a}{b}{\gamma}}{} & = & 
  \{~  \renameaf{a}{b}{t} ~|~  t  \in  \sgamma{\gamma}{E}\} \\
\sd{\antiprojection_{a}({\gamma})}{} & = & 
  \{~  \antiprojection_{a}({t})  ~|~  t  \in  \sgamma{\gamma}{E}  \} \\
 \sd{\mu X.~ \gamma ~ \cup ~ \alpha_{rec}}{} & = & 
  \{  \mu X. t \cup  t_{rec} ~|~ t  \in  \sgamma{\gamma}{E} ~\land~ t_{rec} \in {\salpha{\alpha_{rec}}{E}} \}
\end{array}$\end{footnotesize}\caption{Formal semantics of \mudag{} terms, with one interpretation function for each syntactic construct of Fig.~\ref{fig:syntax_muLQDAG}. }\label{fig:semantics_muLQDAG}
\end{figure}

    A well-formed \mudag{} is a \mudag{} whose interpretation is a set of semantically equivalent terms. %
    For example, Fig.~\ref{fig:well-formed-rlqdag} illustrates a well-formed  \mudag{} capturing two semantically equivalent relational terms obtained before and after join distributivity over union. Fig.~\ref{fig:not-well-formed-rlqdag} illustrates a   \mudag{} which is not well-formed, since its top-level equivalence node contains two subterms that are not semantically equivalent.
    \begin{figure}[H]
      \begin{minipage}[t]{0.45\linewidth}
        \begin{center}
        \includegraphics[width=0.9\linewidth]{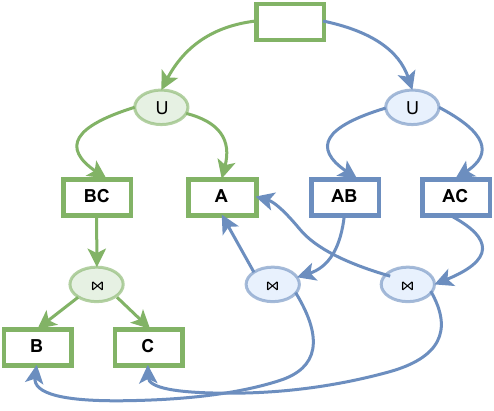}
        \vspace{-0.1cm}
        \caption{Well-formed.}
        \label{fig:well-formed-rlqdag}
        \end{center}
      \end{minipage}
        \hfill    
      \begin{minipage}[t]{0.5\linewidth}
          \begin{center}
        \includegraphics[width=0.8\linewidth]{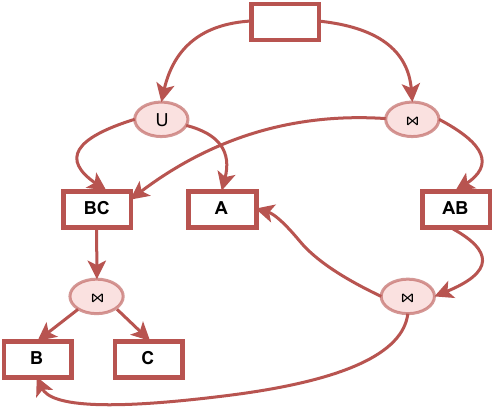}
        \vspace{-0.1cm}
        \caption{Not well-formed.}
        \label{fig:not-well-formed-rlqdag}
          \end{center}
      \end{minipage}
     \end{figure}
   %\   
\noindent The following generalizes well-formedness to encompass recursion: 
    \begin{definrep}[Well-formedness]
    A \mudag{} $\eqnode{\alpha}$ is well-formed if and only if 
    $\forall t, t' \in \salpha{\alpha}{\emptyset}, \quad \evaluation{t} =  \evaluation{t'}$.
    \end{definrep}
    In this definition $\evaluation{t}$ returns the set-semantics interpretation of the individual recursive relational algebraic term $t$ (i.e. the set of tuples returned by $t$ when evaluated in a database instance).

    The  \emph{type} of an operation node $d$ is the set of column names obtained in the result of the evaluation of any subbranch of $d$.
    In a well-formed \mudag{}, all $d_{i}$ under the same equivalence node are semantically equivalent, and thus have the same type.
    For this reason, we also define the \emph{type of an equivalence node}: $\type{\gamma}$ as the type of one of its operation nodes. 
    Notice that the type of an annotated equivalence node of some fixpoint operation node  $d$ corresponds to the type of the equivalence node of the constant part of $d$.  
    For a given filter operation node $\filtaf{f}{{\gamma}}$, we denote by $\filtColumns{f} \subseteq \type{\filtaf{f}{{\gamma}}}$ the subset of column names used in the filtering function $f$.

    \subsection{Recursive terms and rule applicability}
\label{section:rec-terms_rule-applicability}

A significant novelty introduced by recursion resides in the criteria used to trigger rewrite rules. In the non-recursive setting these criteria are trivial in the sense that they only depend on top-level operators. In the example of Fig.~\ref{fig:well-formed-rlqdag}, when applying join distributivity over union, the applicability of the rewrite rule $A \cup (B \bowtie C) ~ \longrightarrow ~ (A \bowtie B) \cup (A \bowtie C)$ can be determined by examining only the combination of the two top-most operators, i.e. the top-level ($\cup$) with the operator immediately underneath (i.e. $\bowtie$).
 
For some other rewrite rules, applicability criteria may also include some additional verifications such as (non)-interaction between e.g. the set of columns being filtered, or the columns being removed (in the cases of filter and antiprojection, respectively).
In the non-recursive setting, these verifications can be done without the need to traverse the whole term. This means that usually no further traversals of subtrees of operators are required. %
For instance, in the previous example of join distributivity over union, $B$ and $C$ do not need to be traversed at all when determining rule applicability. %
In sharp contrast, rules for transforming recursive terms rely on criteria that are significantly more complex as they sometimes require a whole traversal of the recursive part of a fixpoint term. This is because opportunities for rule application with recursive terms depend on how the recursive variable is used within the recursive parts of fixpoints. It is known since the works of \cite{LFP-RA,lfp-ra-extension,geneves-sigmod20} that criteria for rule application are significantly more complex in the presence of recursion as they need to examine how the recursive variables are used. A key contribution of this work is to show that it is still possible to apply rules over sets of recursive terms at once,  using the new concept of \emph{annotated equivalence nodes} (introduced in \S~\ref{section:annotated_alternative}). We first need some preliminary definitions. 
 
\subsection{Preliminary definitions for \mudag{}}
\label{section:criteria-definition}
\newcommand{\unfold}[1]{\texttt{unfold}(#1)}
%
% \vspace{-0.1cm}
\begin{definrep}[Unfolding]
Let $\alpha$ be an equivalence node. The unfolding of $\alpha$, denoted $\unfold{\alpha}$ is $\alpha$ in which all occurrences of equivalence node variable names $Y$ are replaced by their definitions (binders are simply unfolded).
\end{definrep}

We now define two auxiliary functions over \mudag{} terms.
These functions will be used for defining the new concept of annotated equivalence nodes introduced in Section~\ref{section:annotated_alternative}.

\paragraph{Notion of destabilizer in \mudag{}} 
We define the destabilizer of an \mudag{} equivalence node as the set of columns that can be modified by an iteration of a parent fixpoint node. 
Specifically, \dee{}() traverses  subterms and analyzes how the occurrences of free variables are used in order to compute the set of columns that are subject to modifications during a fixpoint node iteration (e.g. renaming or antiprojection).

\begin{definrep}
For a fixpoint operation node $\mu X.~\gamma~ \cup~ \alpha_{rec}$ we consider $\alpha'=\unfold{\alpha_{rec}}$ and we define $\destab(\alpha', X)$ as the following set of column names:
\begin{small}
$\destab(\alpha', X) = \{c \in \mathfrak{C} ~|~ \exists p \in \deriv{\alpha'}{X} ~~p(c) \neq c \}$
\end{small}
\noindent 
where $\mathfrak{C}$ is an infinite set of column names and $\deriv{\cdot}{\cdot}$ computes the set of derivations over a \mudag{} term:%
\begin{center}\vspace{-0.1cm}
   \begin{footnotesize}
   \renewcommand\arraystretch{1.4} 
  $\begin{array}{rcl}
  \deriv{(d, \alpha)}{X} & = & \deriv{d}{X} \\
  \deriv{\eqnode{\alpha_1} \cup \eqnode{\alpha_2}}{X} & = &   \deriv{\alpha_1}{X} \cup \deriv{\alpha_2}{X}  \\
  \deriv{ \eqnode{\alpha_1} \triangleright \eqnode{\alpha_2} }{X} & = & \deriv{\alpha_1}{X}  \\
  \deriv{ \eqnode{\alpha_1} \bowtie \eqnode{\alpha_2}}{X} & = & \deriv{\alpha_1}{X} \cup \deriv{\alpha_2}{X} \\
  \deriv{\renameaf{a}{b}{\eqnode{\alpha}}}{X} & = & \{p \circ (b \rightarrow a, a \rightarrow \perp) \thickspace | \thickspace p \in \deriv{\alpha}{X} \} \\
  \deriv{\antiprojection_{a}{(\eqnode{\alpha})}}{X} & = & \{p \circ ( a \rightarrow \perp) \thickspace | \thickspace p \in \deriv{\alpha}{X} \} \\
  \deriv{\sigma_f(\eqnode{\alpha})}{X} & = & \deriv{\alpha}{X} \\
  \deriv{\mu (Z.~\gamma~ \cup~ \eqnode{\alpha}_{\mathfrak{D}}^{\mathfrak{R}})}{X} & = & \emptyset \\
  \deriv{X}{X} & = & \{ () \} \thickspace $ (a singleton identity)$ \\
  \deriv{R}{X} & = & \emptyset \\
%   \destab(|c \rightarrow c|, X) & = & \emptyset \\
 \end{array}$
\end{footnotesize}
\end{center}\vspace{-0.1cm}
and where $\circ$ represents the composition and $(a_1 \rightarrow b_1, ..., a_n \rightarrow b_n)$ denotes the function that maps each $a_i$ to its $b_i$ and every other column name to itself. 
\end{definrep}

For instance, in the \mudag{} of Fig.~\ref{fig:lqdag_recursion}, $\mathfrak{D} = \{s,m\}$ in $\Sigma$ and $\mathfrak{D'} = \{t,m\}$ in $\Sigma'$. Intuitively, this is because these columns are renamed in front of the recursive variables.

\paragraph{Notion of rigidity in \mudag{}}

We define a function \strict() that computes the set of columns that cannot be added nor removed from a fixpoint operation node, without breaking the semantics of the \mudag{} term. %
 A column $c \in \mathfrak{C}$ cannot be added nor removed from an annotated equivalence node $\alpha_{rec}$ (recursive in $X$)
when $c ~\in~ \strict(\unfold{\alpha}, X)$ and \strict() is defined as follows:

\begin{definrep}[Rigidity]\mbox{ }\\
    \begin{footnotesize}
     \renewcommand\arraystretch{1.4} 
    $\begin{array}{rcl}
      \strict((d, \alpha), X ) & = & \strict(d, X) \\
      \strict(\eqnode{\alpha_1} \cup \eqnode{\alpha_2}, X) & = &  \strict(\alpha_1, X) \cup \strict(\alpha_2, X)  \\
      \strict(\eqnode{\alpha_1} \bowtie \eqnode{\alpha_2}, X) & = & \strict(\alpha_1, X) \cup \strict(\alpha_2, X)  \\
      \strict(\eqnode{\alpha_1} \triangleright \eqnode{\alpha_2}, X) & = &\strict(\alpha_1, X) \cup \strict(\alpha_2, X)  \\
      \strict(\renameaf{a}{b}{\eqnode{\alpha}}, X) & = &  \strict(\alpha, X) \cup \{a, b\}   \\
      \strict(\antiprojection_{a}{(\eqnode{\alpha})}, X) & = & \emptyset~~ \text{when}~ X \notin \text{\emph{free}}(\alpha)\\
      & = & \strict(\alpha, X) \cup \{a\}~\text{otherwise}  \\
      \strict(\sigma_f(\eqnode{\alpha}), X) & = & \strict(\alpha, X) \cup \text{\emph{filt}}(f) \\
      \strict(\mu (Z.\gamma\cup\eqnode{\alpha}_{\mathfrak{D}}^{\mathfrak{R}}), X) & = & \strict(\alpha, X) \cup \strict(\gamma, X) \\
      \strict(R, X) & = & type(R) $ when $ X \neq R \\
      \strict(X, X) & = & \emptyset \\
     \end{array}$

    \end{footnotesize}
    
    \end{definrep}

    For instance, in the \mudag{} of Fig.~\ref{fig:lqdag_recursion}, $\mathfrak{R} = \{s,t\}$ in $\Sigma$. Intuitively, this is because these columns cannot be added nor removed without changing the semantics of the recursion.

    \subsection{Annotated equivalence node}
    \label{section:annotated_alternative}

 An annotated equivalence node ($\alpha_{rec}$ in the abstract syntax of \mudag{} terms given in Fig.~\ref{fig:syntax_muLQDAG}) is an equivalence node of a recursive part of a fixpoint, which is annotated with information that characterize how the recursive variable is used. Specifically:
 
 \begin{definrep}
  Given a \mudag{} operation node $d = \mu X.~ \gamma~ \cup~ \alpha_{rec}$, the annotated equivalence node $\alpha_{rec}$ is defined as:
  \begin{small}
    $\eqnode{\alpha}_{\mathfrak{D}}^{\mathfrak{R}}$
where $\mathfrak{D} = \destab(\alpha, X)$ and $\mathfrak{R} = \strict(\alpha, X)$. 
  \end{small}
 \end{definrep}

 For example, in Fig.~\ref{fig:lqdag_recursion}'s \mudag{}, subterms $\Sigma$ and $\Sigma'$ (that correspond to the right-to-left and left-to-right path traversals) carry different annotations: $\mathfrak{D}=\{s,m\}$ in $\Sigma$ whereas $\mathfrak{D'}=\{t,m\}$ in $\Sigma'$. They will be used for guiding transformations of recursive \mudag{} terms.
 For instance, on $\Sigma$, pushing a filter on $t$ is possible but not on $s$.  A pair ($s$, $t$) might not pass the filter but still be useful to discover a query answer ($s'$, $t$) that passes the filter. %The reciprocal observation holds for the other direction. 
% Annotations differ for that reason: transformations must be applied (1) in a semantics-preserving fashion and (2) each time they are possible.

Annotations are intended to characterize all the subterms of the annotated equivalence node (thanks to definition~\ref{def:consistency}). %
%In practice, $\mathfrak{D}$ and $\mathfrak{R}$ are computed once only and are valid for all the terms grouped under an equivalence node.
%
Annotated equivalence nodes constitute a novel notion whose goal is to guide and maximize the grouped application of transformations, while also maximizing the sharing of common subterms. %
In the sequel, we detail \mudag{}   transformations, by introducing new rewrite rules capable of transforming sets of subterms at once.

\paragraph{Notion of consistency} 
   
    Intuitively, a \mudag{} is \emph{consistent} iff it is well-formed and in addition, for any annotated equivalence node
    $\eqnode{\alpha_2}_{\mathfrak{D}}^{\mathfrak{R}}$, 
    the annotations $\mathfrak{D}$ and $\mathfrak{R}$ are the same for all operation nodes directly underneath (no matter on which subbranch of the equivalence node they are computed, they coincide). More formally:

    \vspace{-0.1cm}
    \begin{definrep}[Consistency] \label{def:consistency}
      A \mudag{} $\eqnode{\alpha}$ is consistent iff:
      \begin{small}
        \vspace{-0.1cm}
      \begin{enumerate}
        \item it is  well-formed;
        \item for all fixpoint operator node $\mu X. \thinspace \gamma \cup \eqnode{\alpha_2}_{\mathfrak{D}}^{\mathfrak{R}}$ occurring in $\alpha$, we have $\consistent{\alpha_2}{X}{\mathfrak{D}}{\mathfrak{R}}$ where:\\
        $\left\{
        \begin{footnotesize}\renewcommand\arraystretch{1.4} 
        \begin{array}{rcl}
          \consistent{\eqnode{d}}{X}{\mathfrak{D}}{\mathfrak{R}} &  \stackrel{\text{def}}{=}  & \destab(d, X) = \mathfrak{D} ~\land~ \strict(d, X)=\mathfrak{R} \\
          \consistent{\eqnode{d,\alpha}}{X}{\mathfrak{D}}{\mathfrak{R}} & \stackrel{\text{def}}{=} & \destab(d, X)=\mathfrak{D} ~\land~ \strict(d, X)=\mathfrak{R} \\ & &
           ~\land ~ \consistent{\eqnode{\alpha}}{X}{\mathfrak{D}}{\mathfrak{R}} \\
        \end{array}
        \end{footnotesize}
        \right.$
      \end{enumerate}
    \end{small}
    \vspace{-0.1cm}
    \end{definrep}
    \noindent

    An example of a consistent \mudag{} is shown in Fig.~\ref{fig:lqdag_recursion} provided $\mathfrak{D}, \mathfrak{R}, \mathfrak{D'}$ and $\mathfrak{R'} $ are correctly computed. An example of an inconsistent \mudag{} would be a variant of Fig.~\ref{fig:lqdag_recursion} with a wrong annotation (e.g. $\mathfrak{D'}=\mathfrak{D}$), later exposing the structure to incorrect transformations. This is because incorrect criteria satisfaction might then result in, for example, wrongly pushing a join operation node inside a fixpoint operation node. This would typically produce a not well-formed \mudag{}, breaking the semantics of the initial query. In the remaining, we only consider consistent \mudag{}s.

    \section{\mudag{} Transformations}
    \label{section: mudag-adapted-rewrite-rules}

    We now propose \mudag{} transformations whose purpose is to efficiently build the space of equivalent recursive plans.

    \mudag{} transformations capture all the most advanced rewritings of recursive algebraic terms found in previous approaches (e.g. \cite{alpha-RA,geneves-sigmod20}). Unlike previous approaches \mudag{} transformations make it possible to systematically group sets of recursive terms and exploit the sharing of common subterms to avoid redundant computations. %
    \mudag{} transformations leverage annotated equivalence nodes to guide the transformation of recursive subterms. They also update the \mudag{} structure with new annotations when needed, in an incremental manner. The incremental aspect for updating annotations is important as it avoids numerous subterm traversals, thus enabling more efficient grouped transformations. %
    For each recursive transformation, we describe which subterms can be shared, how newly generated terms are attached and what happens with the other plans already present in the equivalence node. The creation of new combinations of operation nodes may in turn generate more opportunities for transformations that are also explored.

  \subsection{\mudag{} rewrite rules} \label{sec:transfo}
    We formalize all these ideas by introducing \mudag{} rewrite rules, based on the syntax of \mudag{} terms introduced in Section~\ref{sec:mudag-syntax}.
    Specifically, \mudag{} rewrite rules are formalized as functions that take an equivalence node $\gamma$ and return another equivalence node $\gamma'$ obtained after applying transformations.
  
    \paragraph{Pushing filters into fixpoint operation nodes}
    For pushing filters into sets of recursive terms, we introduce a function $\pf{}$, defined by considering all the syntactic decomposition cases of the input $\gamma$. Fig.~\ref{fig:pushing-filterr} focuses on the two main cases that correspond to potential opportunities of pushing filters, i.e. the cases $\pf{\eqnode{d}}$ and $\pf{\eqnode{d, \alpha}}$ where $d$ filters an equivalence node which contains a recursive subterm. For all the other cases, $\pf{}$ does not reorder operation nodes in the \mudag{} structure but simply traverses it in search for further transformation opportunities underneath\footnote{For instance, two sample cases are the following:    
    $$\begin{array}{lcl} 
      \pf{\mu X.~ \gamma~ \cup~ \recursiveNodeAnot{\alpha}{D}{R}} & = & \mu X.~ \expand{\gamma} ~\cup~ \eqnode{\expand{\alpha}}_{\mathfrak{D}}^{\mathfrak{R}} \\
      \pf{\eqnode{ \eqnode{A} ~\bowtie~ \eqnode{B}}} & = & \expand{\eqnode{A}} ~\bowtie~ \expand{\eqnode{B}} %\\  
     % ...
    \end{array}$$
    where the $\expand{}$ function, defined in Section~\ref{sec:expandfunction}, simply traverses subterms in search for more transformation opportunities. $\pf{}$ is defined similarly for all the other syntactic cases of $\gamma$. Notice that when $\gamma$ is a reference $Y$, there is no need to introduce a new binder, the reference name is used directly.
    }.

    \begin{figure*}   \begin{footnotesize}
    \input{rule-pf}
    \hspace{0.3cm}
    \minipage{0.24\textwidth}%
    \begin{center}
      \includegraphics[width=0.8\linewidth]{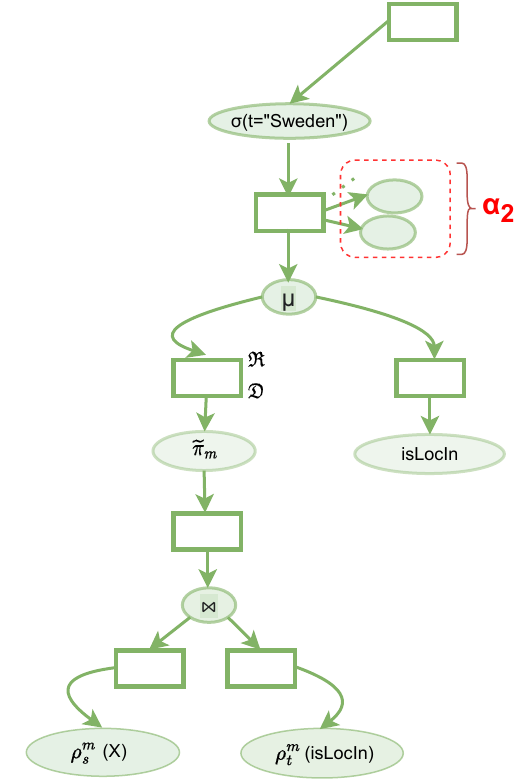}
      \caption{Initial \mudag{} before expansion with opportunity to apply \pf{}.}
      \label{fig:pushing-filter-before}
    \end{center}
  \endminipage
  \hspace{0.1cm}
  \minipage{0.25\textwidth}%
\begin{center}
  \includegraphics[width=1\linewidth]{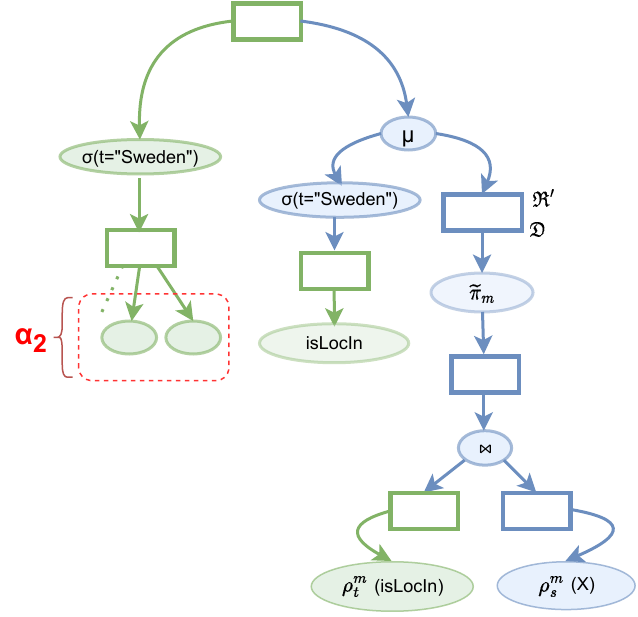} 
      \caption{\mudag{} expansion by pushing a filter in a fixpoint operation node, with branches added by $\pf{\cdot, \text{true}}$ in blue.}
      \label{fig:pushing-filter-after}
\end{center}
  \endminipage
 
\end{footnotesize}
  \end{figure*}

In Fig.~\ref{fig:pushing-filterr}, whenever the filter can be pushed through the fixpoint operation node, $\pf{}$ generates a new \mudag{} subterm ``\pushedf{}'' in which the filter operation node is put within the constant part of the fixpoint operation node. 

 For this transformation to happen, the criteria $\filtColumns{f} \cap \mathfrak{D}$ must be satisfied. Whenever it is not the case, the filter is not rearranged but the \mudag{} structure is simply recursively traversed in search for more transformation opportunities. 
The creation of new combinations of operation nodes may in turn provide new opportunities for other rewritings (for instance, new opportunities for pushing filters even further, or even other kinds of rewritings). This is the role of the $\expand{}$ function, formally defined in section~~\ref{sec:expandfunction}. Intuitively, a call to $\expand{}$ on an equivalence node may further populate the equivalence node with new subterms. The $\expand{}$ function is in charge of exploring all opportunities for transformations. This is useful because other rewrite rules may apply, and the $\expand{}$ function basically triggers all possible  applications of all rewrite rules.  %
Notice that the function $\pf{}$ takes a parameter $\replaceparam{}$ as input. This parameter is used to control whether the initial term ($\unpushedf{}$) is preserved or not in the expansion. For instance, $\pf{\cdot, \text{false}}$ will keep the initial term whereas $\pf{\cdot, \text{true}}$ will replace it by the term  ``$\pushedf{}$'' in which the filter is pushed.  When the parameter $\replaceparam{}$ is omitted, it is assumed to be \text{false}.

For example, the query $\mathcal{Q}_{e_2}$:\begin{center}$\begin{footnotesize}  ?s~~ \leftarrow ~~ ?s~~ \texttt{isLocatedIn+} ~~ \text{Sweden}\end{footnotesize}$\end{center}
 \noindent retrieves all nodes that are connected to the constant node ``\text{Sweden}'' by a path composed of a sequence of edges labeled \begin{small}\texttt{isLocatedIn}\end{small} from the Yago graph dataset \cite{yagooo}. Fig.~\ref{fig:pushing-filter-before} illustrates graphically a portion of $\mathcal{Q}_{e_2}$  \mudag{}, and Fig.~\ref{fig:pushing-filter-after} depicts its updated structure obtained after the pushing filter transformation defined in Fig.~\ref{fig:pushing-filterr}. New branches created by $\pf{}$ are represented in blue color. The new term is added in the same equivalence node as the previous term, since they are semantically equivalent. Notice the incremental update of annotations performed by $\pf{}$ in   Fig.~\ref{fig:pushing-filterr}: the annotations of the newly created term (in blue) are obtained from the annotations of the initial term. In that case $\mathfrak{D}$ is simply propagated whereas $\mathfrak{R'} = \mathfrak{R \cup \filtColumns{f}}$.

      \paragraph{Pushing joins into fixpoint operation nodes}

     For pushing joins into sets of recursive terms we define a function $\pj{}$.  $\pj{}$ takes an equivalence node $\gamma$ as input and returns an expanded equivalence node $\gamma'$ that contains all the subterms in $\gamma$ with, in addition, all the subterms where all joins pushable in fixpoint operation nodes have been pushed. We define $\pj{}$ for all possible syntactic decompositions of a \mudag{}. Fig.~\ref{fig:pushing-joinn} presents the definition of $\pj$ for the two main cases of interest. Again, other cases are defined without structure rearrangement but involving recursive calls to $\expand{}$ in search for further transformation opportunities underneath. %
Notice that whenever a join can be pushed within a fixpoint operation node (see Fig.~\ref{fig:pushing-joinn}), it is possible to share the constant part of the fixpoint operation node. This is made explicit by the creation of the outermost ``let'' binder whose goal is to define and associate a name to the equivalence node, so that it can be referred multiple times (thus explicitly showing the sharing of subterms).

For example, the query $\mathcal{Q}_{e_3}$:\begin{center}$\begin{footnotesize}  ?s,~?t~~ \leftarrow ~~ ?s~~ \texttt{haschild+/livesin} ~~?t\end{footnotesize}$\end{center}
\noindent retrieves all pairs of nodes that are connected by a path composed of a sequence of edges labeled \begin{small}\texttt{haschild}\end{small} followed by a single edge labeled \begin{small}\texttt{livesin}\end{small}. Fig.~\ref{fig:pushing-join}  illustrates graphically a portion of $\mathcal{Q}_{e_3}$  \mudag{}, after the application of the rule that pushes joins into fixpoint operation nodes described in Fig.~\ref{fig:pushing-joinn}. Fig.~\ref{fig:pushing-join} shows that the newly created branch (in blue color) extends the set of semantically equivalent terms of the existing equivalence node.

\begin{figure*}[h]   \begin{footnotesize}
\input{rule-pj}
\minipage{0.48\textwidth}%
\begin{center}
\includegraphics[width=0.8\linewidth]{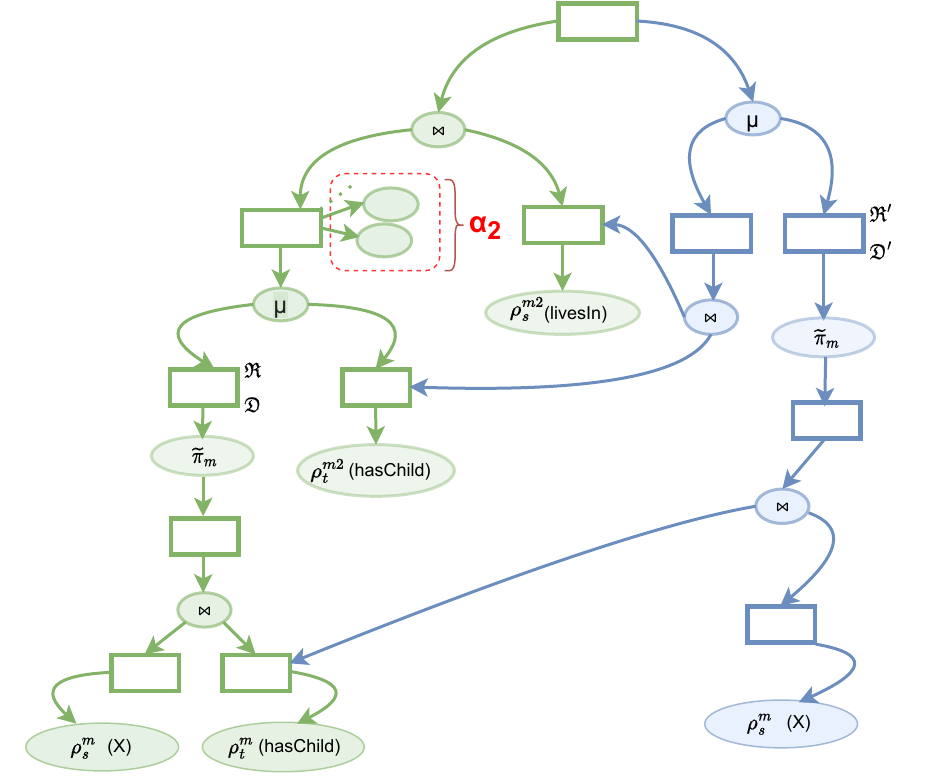}
\end{center}
\caption{Expansion of a \mudag{}  by pushing a join in a fixpoint operation node. 
The initial \mudag{} is in green color, and parts added by $\pj{}$ are in blue color.}
\label{fig:pushing-join}
\endminipage
\end{footnotesize}
\end{figure*}

      \paragraph{Merging fixpoint operation nodes} \label{sec:merging}
      
The function $\mf{}$ defined in Fig.~\ref{fig:merging-fixpointss} takes an input $\gamma$ and returns an equivalence node $\gamma'$ containing all the subterms in $\gamma$ with, in addition, all the terms in which recursions that can be merged are merged. A merging happens whenever (i) two recursions are joined and (ii) their annotated equivalence nodes allow them to be merged into a single recursion, as described in Fig.~\ref{fig:merging-fixpointss}. 
The constant part of the new recursive term created is the join of the constant parts of the two initial fixpoints, and a new recursive part is created. Since the constant part has changed, a new recursive variable is introduced and recursive parts are also new (and cannot be shared\footnote{This does not prevent sharing potential subparts with no occurrence of a free variable.}). %New equivalence nodes are created.    

For example, the query $\mathcal{Q}_{e_4}$: \begin{center}$\begin{footnotesize}  ?s,~?t~~ \leftarrow ~~ ?s~~ \texttt{islocatedin+/dealswith+} ~~?t\end{footnotesize}$\end{center}
\noindent retrieves all pairs of nodes that are connected by a path composed of two successive sequences of edges labeled \begin{small}\texttt{islocatedin}\end{small} and  \begin{small}\texttt{dealswith}\end{small} respectively. Fig.~\ref{fig:merging-fixpoints}  illustrates a portion of $\mathcal{Q}_{e_4}$  \mudag{}, with the new recursive term produced, using subterm-sharing, after the merging of fixpoint operation nodes.

      \begin{figure}
        \begin{footnotesize}
          \minipage{0.45\textwidth}
        \renewcommand\arraystretch{1.1}
        $\mf{ \hm{\bigr[}~ {\eqnode{\mu X_1.~ \gamma_1~ \cup~ \recursiveNodeAnot{\alpha_1}{D_1}{R_1}}} \bowtie {\eqnode{\mu X_2.~ \gamma_2~ \cup~ \recursiveNodeAnot{\alpha_2}{D_2}{R_2}}}  ~\hm{\bigr]}~} =\newline \mbox{ }~\hspace{0cm}\left\{\begin{array}{ll}
          \sbullet[1.1] &  ~\hm{\bigr[}~ \text{let $\text{const}_1$} = \gamma_1~ \text{in}, \text{let $\text{const}_2$} = \gamma_2~ \text{in} \\
            & \quad \expand{\eqnode{ \mu X_1.~ \text{const}_1~ \cup~ \recursiveNodeAnot{\alpha_1}{D_1}{R_1}}} ~\bowtie~ \expand{\eqnode{\mu X_2.~ \text{const}_2~ \cup~ \recursiveNodeAnot{\alpha_2}{D_2}{R_2}}}  ~~ , \\
            & \quad {\color{blue} \mu X.~ \expand{\eqnode{\text{const}_1 \bowtie \text{const}_2}} \cup~ \expand{\eqnode{\subst{\alpha_1}{X_1}{X} ~ \cup ~ \subst{\alpha_2}{X_2}{X}}_{\mathfrak{D}}^{\mathfrak{R}}~ }} ~~ \hm{\bigr]}  \\
          & ~\text{when}~~ \text{$(\type{\gamma_1} \cap \type{\gamma_2}) \cap (\mathfrak{D_1} \cup \mathfrak{D_2}) = \emptyset$} \\
          &  ~\text{and}~ \text{$\type{\gamma_1} \backslash \type{\gamma_2} \cap \mathfrak{R_2} = \emptyset $} ~\text{and}~ \text{$\type{\gamma_2} \backslash \type{\gamma_1} \cap \mathfrak{R_1} = \emptyset $}  \\
         \sbullet[1.1] &\hm{\bigr[}~ \expand{\eqnode{ \mu X_1.~ \text{const}_1~ \cup~ \recursiveNodeAnot{\alpha_1}{D_1}{R_1}}} \bowtie \expand{\eqnode{\mu X_2.~ \text{const}_2~ \cup~ \recursiveNodeAnot{\alpha_2}{D_2}{R_2}}} \hm{\bigr]}~~~\\
         & \text{otherwise}\end{array}
         \right.$
        
        \mbox{}\\
        \mbox{}\\
             % A more general case where an equivalence node can contain more than one operation node:
             $\mf{ \hm{\bigr[}~ {\eqnode{\mu X_1.~ \gamma_1~ \cup~ \recursiveNodeAnot{\alpha_1}{D_1}{R_1}, \alpha_3}} \bowtie {\eqnode{\mu X_2.~ \gamma_2~ \cup~ \recursiveNodeAnot{\alpha_2}{D_2}{R_2}, \alpha_4}}  ~\hm{\bigr]}~} =\newline \mbox{ }~\hspace{0cm}\left\{\begin{array}{ll}
              \sbullet[1.1] &\hm{\bigr[}~ \text{let $\text{const}_1$} = \gamma_1~ \text{in}, \text{let $\text{const}_2$} = \gamma_2~ \text{in} \\
                & \quad \expand{\eqnode{ \mu X_1.~ \text{const}_1~ \cup~ \recursiveNodeAnot{\alpha_1}{D_1}{R_1}}} ~\bowtie~ \expand{\eqnode{\mu X_2.~ \text{const}_2~ \cup~ \recursiveNodeAnot{\alpha_2}{D_2}{R_2}}} ~~ , \\
                & \quad {\color{blue} \mu X.~ \expand{\eqnode{\text{const}_1 \bowtie \text{const}_2}} \cup~ \expand{\eqnode{\subst{\alpha_1}{X_1}{X} ~ \cup ~ \subst{\alpha_2}{X_2}{X}}_{\mathfrak{D}}^{\mathfrak{R}}~ }}~~,\\
                & \quad \expand{\eqnode{\alpha_3}} \bowtie  \expand{\eqnode{\mu X_2.~ \gamma_2~ \cup~ \recursiveNodeAnot{\alpha_2}{D_2}{R_2}, \alpha_4}}  ~~, \\
                & \quad \expand{\eqnode{\alpha_4}} \bowtie \expand{\eqnode{ \mu X_1.~ \gamma_1~ \cup~ \recursiveNodeAnot{\alpha_1}{D_1}{R_1}, \alpha_3} }   ~~ \hm{\bigr]}  \\
              & \text{when}~~ \text{$(\type{\gamma_1} \cap \type{\gamma_2}) \cap (\mathfrak{D_1} \cup \mathfrak{D_2}) = \emptyset$}\\
              & ~~\text{and}~~  \text{$\type{\gamma_1} \backslash \type{\gamma_2} \cap \mathfrak{R_2} = \emptyset $} ~~\text{and}~~ \text{$\type{\gamma_2} \backslash \type{\gamma_1} \cap \mathfrak{R_1} = \emptyset $} \\
             \sbullet[1.1] & \hm{\bigr[}~ \expand{{{\eqnode{\mu X_1.~ \gamma_1~ \cup~ \recursiveNodeAnot{\alpha_1}{D_1}{R_1}, \alpha_3}}}} \bowtie \expand{{\eqnode{\mu X_2.~ \gamma_2~ \cup~ \recursiveNodeAnot{\alpha_2}{D_2}{R_2}, \alpha_4}} }  ~, \\
            & \quad \expand{\eqnode{\alpha_3}} \bowtie  \expand{\eqnode{\mu X_2.~ \gamma_2~ \cup~ \recursiveNodeAnot{\alpha_2}{D_2}{R_2}, \alpha_4}}  ~~, \\
            &  \quad \expand{\eqnode{\alpha_4}} \bowtie \expand{\eqnode{\mu X_1.~ \gamma_1~ \cup~ \recursiveNodeAnot{\alpha_1}{D_1}{R_1}, \alpha_3}}  ~~\hm{\bigr]} ~~ \text{otherwise}  
            \end{array}\right.$
            \mbox{}\\
      
            where $\mathfrak{D} = \mathfrak{D_1} \cup \mathfrak{D_2}$ and $\mathfrak{R}=\mathfrak{R_1} \cup \mathfrak{R_2}$ and $\subst{\alpha_i}{X_i}{X}$ denotes $\alpha_i$ in which all occurrences of $X_i$ are replaced by $X$. This is because the only transformation of the recursive part that needs to be propagated to update the annotations is the union of the two former recursive parts and both $\destab{}$ and $\strict{}$ are distributive over union.
  
          \caption{Merging fixpoint operation nodes.} 
          \label{fig:merging-fixpointss}
          \endminipage\hfill
    \end{footnotesize}
      \end{figure}

      \begin{figure}
       \centering
       \includegraphics[width=1\linewidth]{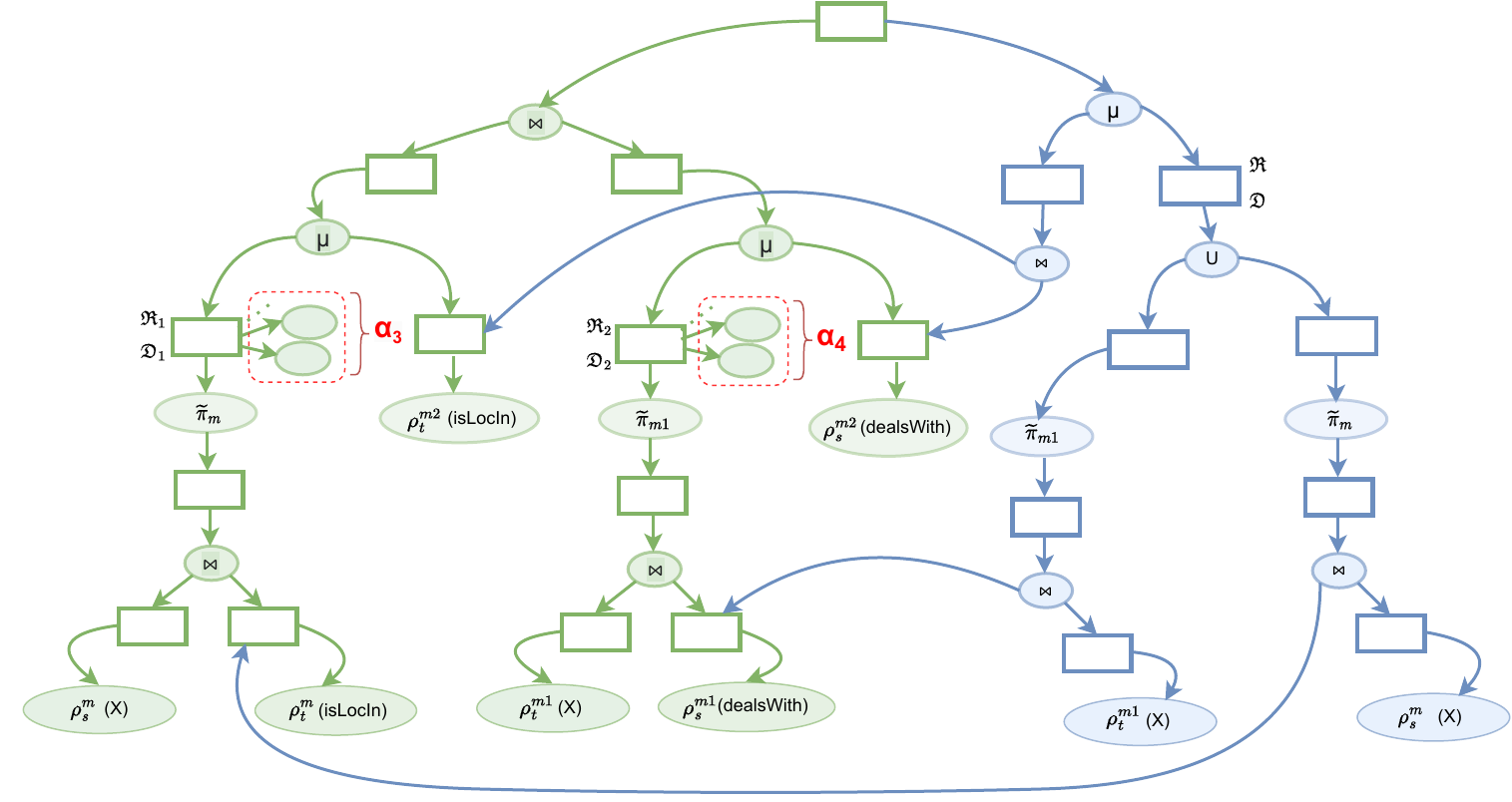}
       \caption{\mudag{} structure after merging recursions.}
       \label{fig:merging-fixpoints}
     \end{figure}

     \paragraph{Pushing an antiprojection in a fixpoint operation node}

\begin{figure}
    \begin{footnotesize}
    \minipage{0.5\textwidth}\renewcommand\arraystretch{1.3} 
    $\pp{\eqnode{~~\antiprojection_a~~{\eqnode{\mu X.~ \gamma~ \cup~ \recursiveNode{\alpha}}}}, \replaceparam{}} =\newline \mbox{ }~\hspace{0.3cm}
    \left\{\begin{array}{ll}
      \sbullet[1.1] & \hm{\bigr[} ~ \unpushedp{}, ~ \pushedp{} ~ \hm{\bigr]}  \text{ when }~~  a \notin \mathfrak{R}  \text{ and \replaceparam{} \text{ is false}}  \\
      \sbullet[1.1] & \hm{\bigr[} ~ \pushedp{} ~ \hm{\bigr]}  \text{ when }~~  a \notin \mathfrak{R}  \text{ and \replaceparam{} \text{ is true}}  \\
     \sbullet[1.1] & \hm{\bigr[} ~ \unpushedp{} ~ \hm{\bigr]} ~~ \text{ otherwise }  
    \end{array}\right.
    $
  %      where $\mathfrak{D'}= \mathfrak{D} \cup \{a\}$ and $\mathfrak{R'}=\mathfrak{R} \cup \{a\}$.  
    \mbox{}\\
    \mbox{}\\
    \mbox{}\\
    % A more general case where an equivalence node can contain more than one operation node:
    $\pp{\eqnode{~~\antiprojection_a~~({\eqnode{\mu X.~ \gamma~ \cup~ \recursiveNode{\alpha}, \alpha_2}})}, \replaceparam{}} =\newline \mbox{ }~\hspace{0.3cm}\left\{\begin{array}{ll}
      \sbullet[1.1] & \hm{\bigr[} ~ \unpushedptwo{}, ~ \pushedp{}, ~ \expandpalphatwo{} ~ \hm{\bigr]}  \text{ when }~~  a \notin \mathfrak{R} \text{ and } \replaceparam{} \text{ is false } \\  
      \sbullet[1.1] & \hm{\bigr[} ~ \pushedp{}, ~ \expandpalphatwo{} ~ \hm{\bigr]}  \text{ when }~~  a \notin \mathfrak{R} \text{ and } \replaceparam{} \text{ is true } \\    
     \sbullet[1.1] & \hm{\bigr[}~ \unpushedptwo{} ~, ~ \expandpalphatwo{} ~ \hm{\bigr]} ~~ \text{ otherwise }  
    \end{array}\right.$
     \\ \text{ where: } \\
$\begin{array}{rcl}
  \unpushedp{} &=& \antiprojection_a(\expand{{\eqnode{\mu X.~ \gamma ~ \cup~ \recursiveNode{\alpha}}}}) \\
  \pushedp{} &=& {\color{blue}\mu X'.~ \text{expand}(\eqnode{\antiprojection_a({\gamma})}) ~ \cup~ \eqnode{\expand{\subst{\alpha}{X}{X'}}}_{\mathfrak{D}}^{\mathfrak{R}} }\\
  \unpushedptwo{} &=& \antiprojection_a(\expand{{\eqnode{\mu X.~ \gamma ~ \cup~ \recursiveNode{\alpha}, \alpha_2}}})\\
  \expandpalphatwo{} &=& \expand{\eqnode{\antiprojection_a(\eqnode{\alpha_2})}}\\
  \end{array}$
      \caption{Pushing antiprojection in an equivalence node containing a fixpoint.}\label{fig:pushing-antiprojectionn}
    \endminipage
    %\hfill
    %\hspace{0.3cm}
    %\minipage{0.2\textwidth}%
    %  \includegraphics[width=1\linewidth]{figures/pushing-antiprojection-before.pdf}
    %  \caption{Initial \mudag{} before expansion with opportunity to apply \pp{}.}
    %  \label{fig:pushing-antiprojection-before}
  %\endminipage
  %\hspace{0.1cm}
  %\minipage{0.24\textwidth}%
  %    \includegraphics[width=1\linewidth]{figures/pushing-antiprojection-after.pdf} 
  %    \caption{Expansion of a \mudag{}  by pushing an antiprojection in a fixpoint operation node. The part added by $\pp{}$ is in blue color.}
  %    \label{fig:pushing-antiprojection-after}
  %   \endminipage      
\end{footnotesize}
  \end{figure}

  %  \begin{figure}[H]
  %   \centering
  %   \includegraphics[width=0.9\linewidth]{figures/pushing-antiprojection.pdf}
  %   \caption{Pushing an antiprojection in a fixpoint.}
  %   \label{fig:pushing-antiprojection}
  % \end{figure}

    For pushing antiprojections into fixpoint operation nodes we introduce a function $\pp{}$ that takes an equivalence node $\gamma$ as input and returns an expanded equivalence node $\gamma'$ where all pushable antiprojections have been pushed. 
    $\pp{}$ is defined in Fig.~\ref{fig:pushing-antiprojectionn}. %
    The antiprojection is pushed when the criteria are satisfied, and this results in the creation of a new term in the equivalence node. Depending on the value of the parameter \replaceparam{} (false when omitted), the initial term is either preserved or discarded in the expansion. 
    Whenever the criteria is not satisfied, the initial term is left unchanged, but traversed in search for more transformation opportunities. %
        %
%Fig.~\ref{fig:pushing-antiprojection-before} illustrates graphically the \mudag{} before the transformations. Fig.~\ref{fig:pushing-antiprojection-after} shows the graphical representation of the \mudag{} obtained by the transformation of Fig.~\ref{fig:pushing-antiprojectionn}. 

    \paragraph{Pushing an antijoin in a fixpoint operation node}

      For pushing antijoins into fixpoints we introduce a function $\pa{}$ that takes an equivalence node $\gamma$ as input and returns an expanded node $\gamma'$ that contains all the subterms in $\gamma$ with, in addition, all subterms where all pushable antijoins have been pushed in fixpoint operation nodes.
      $\pa{}$ is defined in Fig.~\ref{fig:pushing-antijoinn}. Again, Fig.~\ref{fig:pushing-antijoinn} focuses on the syntactic cases that correspond to when an antijoin might be pushed in a fixpoint. When criteria are satisfied, the antijoin is pushed in the constant part of the fixpoint operation node and the newly created term is added to the set of equivalent terms along with the rest.

      \paragraph{Transformations of non-regular expressions}
      Notice that \mudag{} transformations support expressions with non-regular recursive patterns. For instance,  $A^n B^n$ can simply be written as the \mudag{} fixpoint operation node  $\mu X.  [A/B] \cup [A/X/B]$ where ``$/$'' is an abbreviation: $W/Z$ stands for the \mudag{} subexpression that joins the target column of $W$ with the source column of $Z$. As such, transformations are equally applicable for regular and non-regular  fixpoint operation nodes.

      \begin{figure}
        \begin{footnotesize} 
        \minipage{0.5\textwidth}\renewcommand\arraystretch{1.2} 
        $\pa{\eqnode{{\eqnode{\mu X.~ \gamma~ \cup~ \eqnode{\alpha}_{\mathfrak{D}}^{\mathfrak{R}}}}\triangleright \beta}}  =\newline \mbox{ }~\hspace{0.1cm}\left\{\begin{array}{ll}
          \sbullet[1.1]&\text{let const} = \gamma ~\text{in}~ \\
            & \quad \hm{\bigr[} ~~ \eqnode{\expand{\eqnode{\mu X.~ \text{const} ~ \cup~ \eqnode{\alpha}_{\mathfrak{D}}^{\mathfrak{R}}} } \triangleright \expand{\beta}} ~, \\
            & \quad {\color{blue}\mu X'.~ \expand{\eqnode{\text{const} \triangleright \beta}} ~ \cup~  \eqnode{\expand{\subst{\alpha}{X}{X'}}}_{\mathfrak{D}}^{\mathfrak{R}} } ~~ \hm{\bigr]}  \\
          & {when}~~ \type{\beta} \cap \mathfrak{D} = \emptyset \\
         \sbullet[1.1] & \hm{\bigr[} ~~ \expand{\eqnode{\mu X.~ \gamma ~ \cup~ \eqnode{\alpha}_{\mathfrak{D}}^{\mathfrak{R}}}} \triangleright \expand{\beta}~ ~~\bigr] ~~~ {otherwise}  
        \end{array}\right.$
        \mbox{}\\
        \mbox{}\\
        % A more general case where an equivalence node can contain more than one operation node:
        $\pa{\eqnode{{\eqnode{\mu X.~ \gamma~ \cup~ \eqnode{\alpha}_{\mathfrak{D}}^{\mathfrak{R}},~\alpha_2}} \triangleright \beta}}  =\newline \mbox{ }~\hspace{0.1cm}\left\{\begin{array}{ll}
          \sbullet[1.1] & \text{let const} = \gamma ~\text{in}~ \\
            & \quad \hm{\bigr[} ~~ \eqnode{\expand{\eqnode{\mu X.~ \text{const} ~ \cup~  \eqnode{\alpha}_{\mathfrak{D}}^{\mathfrak{R}},~\alpha_2}} \triangleright \expand{\beta}}, \\
            & \quad { \color{blue}\eqnode{\mu X'.~ \expand{\eqnode{\text{const} \triangleright \beta}} ~ \cup~  \eqnode{\expand{\subst{\alpha}{X}{X'}}}_{\mathfrak{D}}^{\mathfrak{R}} } }, \\
            & \quad \eqnode{ \expand{\eqnode{\alpha_2} } \triangleright \expand{\beta} } ~~ \hm{\bigr]}  \\
          & {when}~~ \type{\beta} \cap \mathfrak{D} = \emptyset \\ 
         \sbullet[1.1] & \hm{\bigr[} ~~  \expand{\eqnode{\mu X.~ \gamma ~ \cup~ \recursiveNode{\alpha},~\alpha_2}} \triangleright \expand{\beta}, \\
         & \quad \eqnode{\expand{\eqnode{\alpha_2}}\triangleright \expand{\beta}} ~~ \hm{\bigr]} \qquad {otherwise}  
        \end{array}\right.$

          \caption{Pushing antijoin in an equivalence node containing a fixpoint.}\label{fig:pushing-antijoinn}
        \endminipage
        %\hfill
        %\minipage{0.4\textwidth}%
        %\includegraphics[width=1\linewidth]{figures/pushing-antijoin.pdf}
        %\caption{Expansion of a \mudag{}  by pushing an antijoin in a fixpoint operation node. 
        %The initial \mudag{} is in black color, and parts added by $\pa{}$ are in blue color.}
        %\label{fig:pushing-antijoin}
        % 
      %\endminipage
      \end{footnotesize} 
      \end{figure} 

    \subsection{The overall expansion algorithm}
    \label{sec:expandfunction}
    We can now describe the overall \mudag{}  expansion algorithm. We define the function 
    $\expand{}$ that takes an equivalence node $\gamma$ and returns the equivalence node $\gamma'$ containing all the terms obtained by transformations. %    
    The  $\expand{}$ function is defined as follows:
      \renewcommand\arraystretch{1.4} 
    $$\begin{footnotesize}\begin{array}{rcl}
      \expand{\eqnode{d}} & = & \applyAll{\eqnode{d}} \\
      \expand{\eqnode{d, \alpha}} & = & \applyAll{\eqnode{d}} \cup \expand{\alpha} \\
    \end{array}\end{footnotesize}$$
     where $\applyAll{}$ is in charge of applying all possible transformations  on each operation node. This includes applying all rewrite rules defined in Section~\ref{section: mudag-adapted-rewrite-rules} in combination with the more classical ones of relational algebra%\footnote{    As a slight discrepancy between the theory and the implementation, the implementation of the  \mudag{} expansion avoids some redundant calls to the $\expand{}$ function by implementing the $\applyAll{}$ function in an imperative-style double nested for loop, so as to implement a single case analysis and to trigger $\expand{}$ once only when needed.}
     : 
      \renewcommand\arraystretch{1.4} 
     $$\begin{footnotesize}\begin{array}{rcl}
        \applyAll{\eqnode{d}} & = &  \pf{\eqnode{d}} \cup \pa{\eqnode{d}} \cup \pj{\eqnode{d}} \cup \mf{\eqnode{d}} \\
        && \cup~ \pp{\eqnode{d}} \cup \allCodd{\eqnode{d}}  \\
      \end{array}\end{footnotesize}$$
    where $\allCodd{}$ applies all rewrite rules concerning classical (non-recursive) relational algebra adapted for \mudag{}. For example: $\pfj{}$ for pushing filters in join operation nodes, $\paj{}$ for pushing antiprojections in a join operation node, $\jassoc{}$ for join associativity, $\dju{}$ for distributivity of join over unions, etc.
      \renewcommand\arraystretch{1.4} 
    $$\begin{footnotesize}\begin{array}{rcl}
      \allCodd{\eqnode{d}} & = & \pfj{\eqnode{d}} \cup \paj{\eqnode{d}} \cup \jassoc{\eqnode{d}} \cup ... \cup \dju{\eqnode{d}} 
     \end{array}\end{footnotesize}$$
    For instance, $\pfj{}$ is defined as shown in Fig.~\ref{fig:pushing-filter-joinn-2} for an \mudag{} in which a filter precedes an equivalence node which contains a join operation node. In other cases, $\pfj{}$ recursively traverses the structure with appropriate calls to $\expand{}$ in search for further transformation opportunities. Other rewrite rules of non-recursive relational algebra are also implemented in a similar way. % in the \mudag{}.

  \begin{figure}
  \begin{footnotesize} 
    \minipage{0.5\textwidth}%
    \newcommand{\shortexpand}[1]{\texttt{exp}({#1})}
    \renewcommand\arraystretch{1.4 } 
    $\pfj{\eqnode{\filtaf{f}{\eqnode{\gamma_1 \bowtie \gamma_2, \alpha}}}}  =\newline \mbox{ }~\hspace{0.4cm}\left\{\begin{array}{ll}
      \sbullet[1.1] & \hm{\bigr[} ~ {\color{black}{\shortexpand{\eqnode{\filtaf{f}{\gamma_1}}} \bowtie \shortexpand{\gamma_2}},~\shortexpand{\eqnode{\filtaf{f}{\alpha}}} } ~ \hm{\bigr]}  \\
      & {when}~~  \filtColumns{f} \subseteq \type{\gamma_1} \land  \filtColumns{f} \nsubseteq \type{\gamma_2}  \\ 
     \sbullet[1.1] & \hm{\bigr[} ~ {\color{black}{\shortexpand{\gamma_1} \bowtie \shortexpand{\eqnode{\filtaf{f}{\gamma_2}}}},~\shortexpand{\eqnode{\filtaf{f}{\alpha}}} } ~ \hm{\bigr]}  \\
     & {when}~  \filtColumns{f} \nsubseteq \type{\gamma_1} \land  \filtColumns{f} \subseteq \type{\gamma_2}  \\ 
     \sbullet[1.1] & \hm{\bigr[} ~ {\color{black}{\shortexpand{\eqnode{\filtaf{f}{\gamma_1}}} \bowtie \shortexpand{\eqnode{\filtaf{f}{\gamma_2}}}},~\shortexpand{\eqnode{\filtaf{f}{\alpha}}} } ~ \hm{\bigr]}  \\
     & {when}~  \filtColumns{f} \subseteq \type{\gamma_1} \land  \filtColumns{f} \subseteq \type{\gamma_2}  \\ 

    \sbullet[1.1] & \hm{\bigr[} ~ \filtaf{f}{\shortexpand{\eqnode{\phi \bowtie \psi}},~\shortexpand{\eqnode{\filtaf{f}{\alpha}}}} ~ \hm{\bigr]}  ~~ otherwise
    \end{array}\right.$
    \caption{Pushing a filter in an equivalence node composed of at least one join operation node and other operation nodes (\shortexpand{} stands for \expand{}).}
    \label{fig:pushing-filter-joinn-2}
  \endminipage
  \end{footnotesize} 
\end{figure}

    \subsection{Correctness and completeness}
    \label{section:properties}

\begin{proprep}[Correctness]
    Let $\eqnode{\alpha}$ be a consistent \mudag{}, and $\alpha' = \expand{\eqnode{\alpha}}$, then we have $\sgamma{{\eqnode{\alpha}}}{\emptyset} \subseteq \sgamma{\alpha'}{\emptyset}$ and  $\alpha'$ is consistent. 
    \end{proprep}

    \begin{proof}[Proof]
        The proof is done by decomposition and induction. It is available in Appendix~A of the extended version\footnote{\url{https://doi.org/10.48550/arXiv.2312.02572}}.
    \end{proof}

     \begin{proprep}[Completeness properties]
        Let $R$ be a set of \mudag{} rewrite rules such that $R$ contains the 5 \mudag{} rewrite rules for recursive terms presented in Section~\ref{section: mudag-adapted-rewrite-rules}, we consider $\eqnode{\alpha}=\unfold{\expandR{R}{\eqnode{\alpha'}}}$ where $\alpha'$ is a consistent \mudag{}, and $\expandR{R}{}$ is the $\expand{}$ function in which rules in $R$ are activated. 
        The following properties hold:

        \begin{proptyrep}[All pushable filters have been pushed]~\\
            \begin{footnotesize}
            $\forall~ \filtaf{f}{\gamma} \in \eqnode{\alpha}$, $\nexists{} d \in \gamma \mid d=\mu X. ~ \eqnode{\kappa} \cup \eqnode{\alpha_2}_{\mathfrak{D}}^{\mathfrak{R}}$ and $\filtColumns{f} \cap \mathfrak{D} = \emptyset$. 
            \end{footnotesize}
            \end{proptyrep}
      %      \vspace{0.05cm}
            \begin{proptyrep}[All pushable antiprojections pushed]~\\
                \begin{footnotesize}
                $\forall~ \antiprojection_{a}(\gamma) \in \eqnode{\alpha}$, $\nexists{} d \in \gamma \mid d=\mu X. ~ \eqnode{\kappa} \cup \eqnode{\alpha_2}_{\mathfrak{D}}^{\mathfrak{R}}$ and $a \not\in \mathfrak{R}$.
                \end{footnotesize}
            \end{proptyrep}
       %     \vspace{0.05cm}
            \begin{proptyrep}[All pushable joins pushed]~\\
                \begin{footnotesize}
                $\forall~ ( {\beta} \bowtie {\gamma}) \in  \eqnode{\alpha}$, if $\exists d\in\gamma$ such that $d=\mu X. ~ \eqnode{\kappa} \cup \eqnode{\alpha_2}_{\mathfrak{D}}^{\mathfrak{R}}$ and $\type{\beta} \cap \mathfrak{D} = \emptyset$ and $\type{\beta} \backslash \type{\gamma} \thickspace  \cap  \mathfrak{R} = \emptyset$ then $\exists d' \in \eqnode{\alpha}$ such that $d'=\mu X'. ~ \eqnode{ \eqnode{\beta} \bowtie {\eqnode{\kappa}} } \cup \eqnode{\subst{\alpha_2}{X}{X'}}_{\mathfrak{D'}}^{\mathfrak{R'}}$.
                \end{footnotesize}
                \end{proptyrep}
        %        \vspace{0.05cm}
                \begin{proptyrep}[All mergeable fixpoints merged]~\\
                    \begin{footnotesize}
                    $\forall~ ( {\gamma_1} \bowtie {\gamma_2}) \in  \eqnode{\alpha}$, if 
                     $\mu X_1. ~ \eqnode{\kappa_1} \cup \recursiveNodeAnot{\alpha_1}{D_1}{R_1} \in \gamma_1$ and $\mu X_2. ~ \eqnode{\kappa_2} \cup \recursiveNodeAnot{\alpha_2}{D_2}{R_2} \in \gamma_2$ and we have 
                     $\gamma_1 \neq \gamma_2$ and $(\type{\gamma_1} \cap \type{\gamma_2}) \cap (\mathfrak{D_1} \cup \mathfrak{D_2}) = \emptyset$  and  $\type{\gamma_1} \backslash \type{\gamma_2} \cap \mathfrak{R_2} = \emptyset $ and $\type{\gamma_2} \backslash \type{\gamma_1} \cap \mathfrak{R_1} = \emptyset $
                     then there exists $d \in  \eqnode{\alpha}$ such that $d=\mu X. ~ \eqnode{ \eqnode{\kappa_1} \bowtie {\eqnode{\kappa_2}} } \cup \recursiveNodeAnot{ \eqnode{\alpha_1} \cup \eqnode{\alpha_2}}{D}{R}$.
                    \end{footnotesize}
                \end{proptyrep}
    %\vspace{0.05cm}
                \begin{proptyrep}[All pushable antijoins pushed]~\\
                    \begin{footnotesize}
                    $\forall~ (\eqnode{\mu X.~ \gamma~ \cup~ \eqnode{\alpha}_{\mathfrak{D}}^{\mathfrak{R}}} \triangleright \beta) \in \eqnode{\alpha}$,
    if $\type{\beta} \cap \mathfrak{D} = \emptyset$ then $\exists d\in \eqnode{\alpha}$ such that $d=\eqnode{\mu X'.~ \eqnode{\gamma \triangleright \beta} ~ \cup~  \eqnode{\subst{\alpha}{X}{X'}}_{\mathfrak{D}}^{\mathfrak{R'}}}$
                    \end{footnotesize}
                \end{proptyrep}
        \end{proprep}

    \begin{proof}[Proof Sketch]
    These properties are proved by contradiction: (i) assuming the existence of a missed transformation opportunity in the expansion, which (ii) necessarily implies some unrealized rule application (whereas the rule was applicable), and (iii) showing that the systematic structure traversal performed by $\expandR{R}{}$ leaves no room for such a missed opportunity, thus (iv) leading to a contradiction.  
    \end{proof}

\subsection{Implementation}

We implemented the \mudag{} in Scala. It takes as input a recursive query, generates its plan space, selects a plan with best estimated cost, and then send it to PostgreSQL for evaluation. 

The system implementation is organized in two modules: the \texttt{enumerator} that computes the \mudag{} expansion, and the \texttt{cost estimator} in charge of (i) annotating equivalence nodes with estimated cardinalities (using base relations data statistics), and (ii) operation nodes with estimated costs (as described in \cite{lawal:tel-03322720,lawal}).

Our experiments focus on the evaluation of the efficiency of plan enumeration and the quality of the generated plan spaces. Therefore, we first expand, and then estimate costs.

Plan exploration is achieved by an implementation of the $\expand{}$ function presented in Section~\ref{sec:expandfunction}. This implementation relies on a top-down traversal of the \mudag{} structure using two nested loops and recursive calls to $\expand{}$. The purpose of the loops is to efficiently traverse the \mudag{} structure in order to trigger a rule whenever it is applicable. In order to test rule applicability, we use Scala's pattern matching on the structure exposed by the nested loops. There is thus a direct correspondence between the implementation and the formal presentation of the rules that use syntactic case by case decompositions. In the implementation, the application of rules follows an order determined by the depth-first exploration of the \mudag{}'s structure and the implicit order of operation nodes within equivalent nodes (ordered sets in Scala). 
In-memory unicity of operation and equivalence nodes is ensured. At each step of the expansion, equivalence nodes are unified as early as possible, so that duplicates are never created.
Without timeout, since  \mudag{} terms are finite and all recursive calls are done on strictly smaller subterms, the  expansion always terminates. 
Furthermore, all considered rules are deterministic and all applicable rules are always applied, in all possible sequences of application.
Therefore, for a given \mudag{}, the result of the expansion is always the same: the expansion is deterministic. 
In general, filtering expressions (predicates) can be arbitrarily complex. Filter pushdown thus requires a cost estimation in order to decide which is the best plan between the initial and the rewritten term. This is why, by default, rules $\pf{}$ and $\pp{}$ always preserve the initial term in the expansion (i.e. by default, $\replaceparam=$ false).

When filtering expressions are very simple (e.g. they can be evaluated in constant time), it makes sense to  consider only terms in which filtering expressions are pushed.
This simple heuristic (enabled by  setting $\replaceparam=$ true), used in experiments, allows to discard suboptimal plans during plan space exploration.

\section{Experimental Results}

\label{section:experimental_results}

We evaluate the \mudag{} experimentally. Our assessment is driven by the following research questions:

\begin{itemize}
  \item[\textbf{RQ1}]  How efficient is \mudag{} exploration of recursive plan spaces compared to the state-of-the-art? 
  
  \item[\textbf{RQ2}] How relevant are explorations of large recursive plan spaces for practical query evaluation?

\end{itemize}

In the following we first introduce the experimental setup and the experimental methodology with  chosen baselines and metrics. We then report on the results and the main lessons learned.

\subsection{Experimental setup}

\paragraph{Datasets} We consider various unmodified third-party datasets, graphs and trees, real and generated, as described in Table~\ref{table:datasets}.

\begin{small}
\renewcommand\arraystretch{1.1} 
\begin{table}[h]
  \begin{center}
    \begin{tabular}{|r|c|c|c|}  
     \hline
     Dataset & \#nodes \& \#edges & Type & Nature \\ [0.5ex] 
     \hline\hline
     Yago~\cite{yagooo} & 42M \& 62M & Knowledge graph & Real \\ 
     \hline
     Bahamas Leaks~\cite{bahamas-leaks} & 202K \& 249K & Property graph & Real \\
     \hline 
     Airbnb~\cite{chandan-sharma-airbnb,airbnb} & 24K \& 14K & Property graph & Real \\
     \hline
     LDBC~\cite{LDBC} & 908K \& 1.9M  & Property graph & Synthetic\\
     \hline
     Wikitree~\cite{wikitree} & 1.3M \& 9.1M & Tree & Real\\
     \hline
     Academic tree~\cite{academicTree} & 765K \& 1.5M  & Tree & Real\\
     \hline
    \end{tabular}
    \end{center} 
  \caption{Datasets (available from \cite{mulqdagpublicrepo}).}
  \label{table:datasets}
  \vspace{-0.9cm}
  \end{table} 
\end{small}
\paragraph{Queries} 
We consider a variety of recursive queries formulated against these datasets. Queries for Yago are mainly third-party regular path queries already considered in earlier papers in the literature\footnote{We consider 7 queries (Q1-Q7) taken from~\cite{tasweet-2017}, 2 queries (Q8-Q9) taken from~\cite{waveguide-2015},  (Q10-Q11) taken from~\cite{yago-queries-paper3} and (Q12-Q20) come from \cite{geneves-sigmod20}.}, and chosen because they are representative of the variety of possible recursive optimizations that can apply to them. Queries over the Airbnb dataset are inspired from \cite{chandan-sharma-airbnb}. We added more queries formulated over the Bahamas and LDBC datasets. We consider non-regular queries (variants of same generation, and a$^n$b$^n$) for the Wikitree and Academic Tree datasets. Queries and datasets used in experiments are available at \cite{mulqdagpublicrepo}.

\paragraph{Hardware and environment}
All experiments are conducted on a machine with an Intel Xeon Silver 4114 2.20GHz CPU, 192GB of RAM and 4TB 7200rpm SATA disks. The machine runs Ubuntu 20.4, Java 8 and the Scala 2.11 compiler with default parameters.

\subsection{Efficiency of plan space exploration}

To answer \textbf{RQ1}, we implemented a prototype of the \mudag{} and compare it with \muenum{}, which is the plan enumerator of the state-of-the-art $\mu$-RA system \cite{geneves-sigmod20}. As described in Section~\ref{section-related_work}, this system is of the most advanced relational-based system for recursive query optimization; providing the richest plan spaces for recursive terms. \muenum{} \cite{geneves-sigmod20} explores them using a state-of-the-art dynamic programming strategy, in which terms are made unique in memory so as to obtain very efficient term equality tests. 

 We measure the enumeration capability of the \mudag{} in terms of the number of plans explored per second for each query. We compare the performance of the \mudag{} implementation with the performance of \muenum{}. Figures~\ref{fig:yago-log}, \ref{fig:bahamas-log}, \ref{fig:airbnb-ldbc-log}, and Fig.~\ref{fig:non-reg-queries} respectively show the results obtained for the queries over each dataset. %
 \begin{figure}[h]
  \centering
    \includegraphics[width=0.9\linewidth,trim={0cm 0.8cm 0 0},clip]{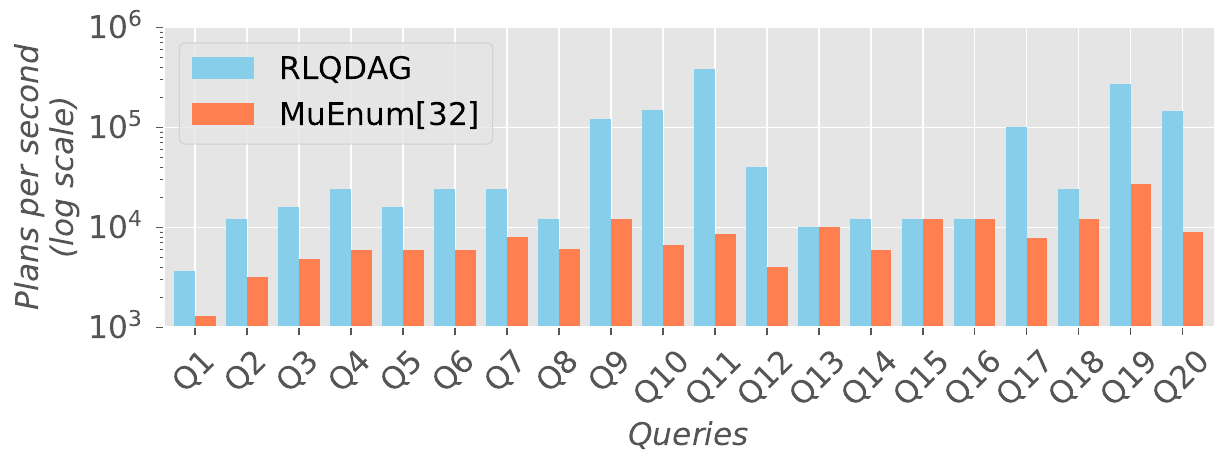}
   \vspace{-0.31cm}
   \caption{Yago queries.}\label{fig:yago-log}
   \vspace{-0.4cm}
\end{figure}
\begin{figure}[h]
  \centering
    \includegraphics[width=0.9\linewidth,trim={0cm 0.8cm 0 0},clip]{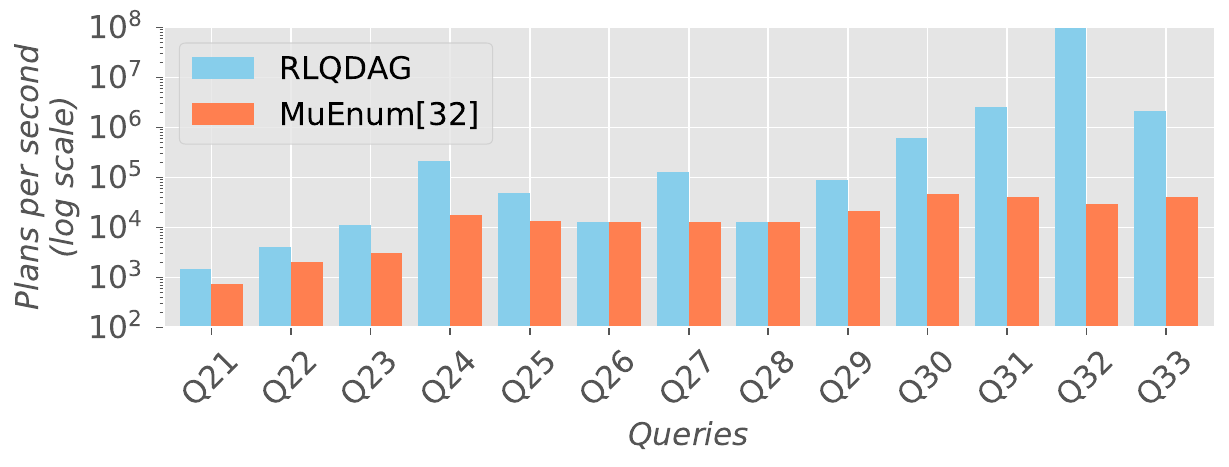}
    \vspace{-0.305cm}
    \caption{Bahamas Leaks queries.}\label{fig:bahamas-log}
    \vspace{-0.6cm}
  \end{figure}
  
\begin{figure}[h]
  \centering
% [trim={left bottom right top},clip]
\includegraphics[width=0.9\linewidth,trim={0cm 1cm 0 0},clip]{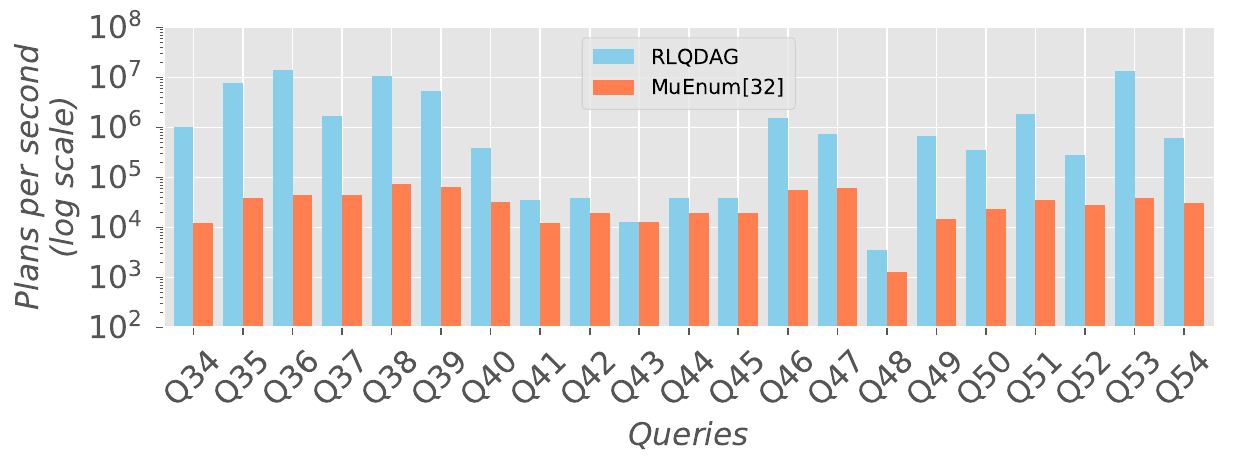}
\vspace{-0.3cm}
\caption{Airbnb and LDBC queries.}\label{fig:airbnb-ldbc-log}
\vspace{-0.4cm}
\end{figure}

\begin{figure}
  \centering
  % [trim={left bottom right top},clip]
    \includegraphics[width=0.9\linewidth,trim={0cm 0.75cm 0 0},clip]{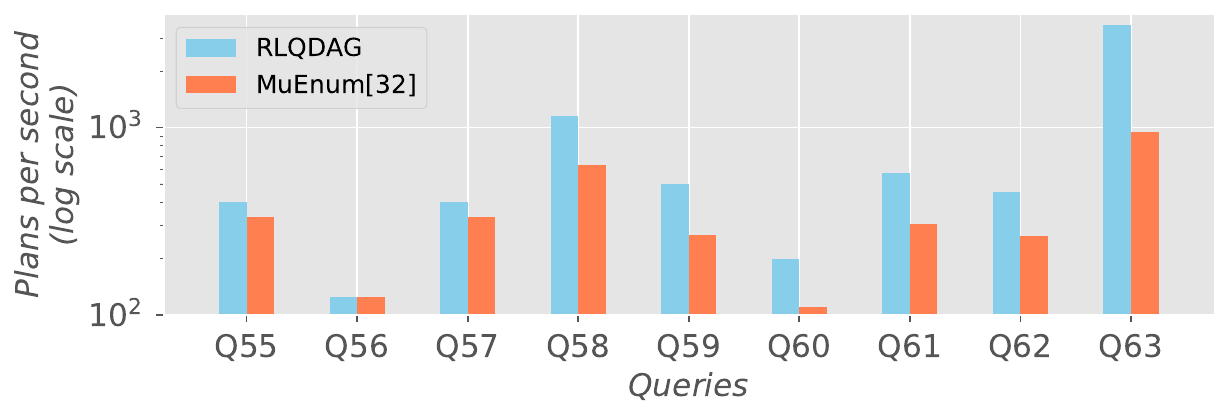}
    \vspace{-0.3cm}
    \caption{Non-regular queries.}\label{fig:non-reg-queries}
    \vspace{-0.4cm}
\end{figure}
In these figures, the $y$ axis (in log scale) indicates the number of plans per second explored by each approach for a given query (on the $x$ axis).
Results suggest that the \mudag{} approach always enumerates plans much faster (up to two orders of magnitude) when compared to \muenum{}\footnote{Cases where histogram bars look very similar correspond to situations where both systems completed the plan space exploration in a very short time duration, which makes the difference hardly visible with the log scale.}.

Now, we set a time budget $t$ (in seconds) for the plan space exploration and let the two systems generate plan spaces for that time budget. This means that after $t$ elapsed seconds we stop the two explorations and look at the plan spaces obtained by the two systems. Fig.~\ref{fig:stacked-plans}  shows the results obtained for a time budget of $t=0.5$ seconds with queries from Bahamas Leaks (Q31-Q32), Airbnb (Q34-Q37-Q38-Q39) and LDBC (Q51-Q53-Q54). The $y$ axis (in log scale) indicates the number of plans found. Fig.~\ref{fig:stacked-plans} also indicates the size of the complete (exhaustive) plan space obtained without any time restriction for the exploration ($t=\infty$). For example, for query Q31, the complete plan space contains more than 21.4 million plans. In 0.5 seconds, the \mudag{} prototype explored 1,019,026 plans whereas \muenum{} explored only 5,751 plans. This is because although \muenum{} uses dynamic programming techniques, it is not capable of benefitting from the \mudag{}'s grouping effect when applying complex rewrite rules on sets of recursive terms at once, thus rules are significantly more costly to apply. Seen from another perspective, this means that the \mudag{}'s approach is more effective in avoiding redundant subcomputations. We have conducted extensive experiments and overall results indicate that, for a given time budget, the \mudag{} prototype explores many more terms in comparison to \muenum{}, in all cases. In some cases, the \mudag{} generates a number of plans which is greater by up to two orders of magnitude (for the same time budget). Such a speedup sometimes enables a complete exploration of the whole plan space in some cases (as shown e.g. for Q34, Q51, Q53, Q54 in Fig.~\ref{fig:stacked-plans}).  

\begin{figure}[H] 
  \centering  \vspace{-0.2cm}
  \includegraphics[width=0.9\linewidth]{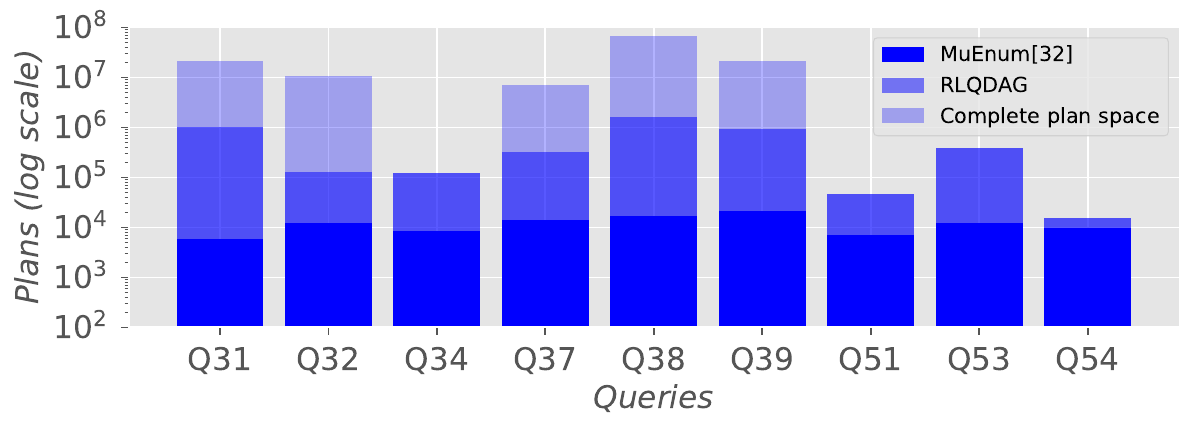} 
  \vspace{-0.3cm}
  \caption{Plan spaces explored in 500ms.}
  \label{fig:stacked-plans} 
  \vspace{-0.3cm}
\end{figure}

We now report on experiments of exploring plan spaces with varying and increasing time budgets for the same query. 
For instance, Fig.~\ref{subfig:q31-q53-different-times} presents the number of plans explored (on the $y$ axis) for different time budgets shown on the $x$ axis for query Q31 and query Q53 (2 queries from 2 different datasets). 
\begin{figure}[H]
  \centering
    \includegraphics[width=0.9\linewidth]{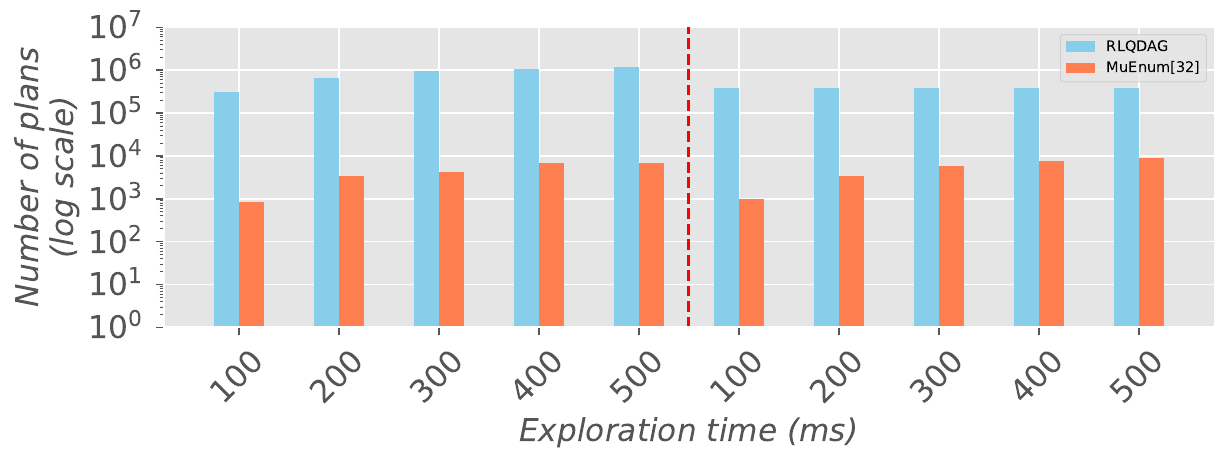}
    \vspace{-0.35cm}
    \caption{Plans explored per time budgets for Q31 and Q53.}\label{subfig:q31-q53-different-times}
    \vspace{-0.1cm}
\end{figure} 
Again, results shown in Fig.~\ref{subfig:q31-q53-different-times} indicate that the \mudag{} explores significantly more plans than the other approach for all considered time budgets (the $y$ axis is in log scale). We can also observe that the difference between the amount of plans explored by each system stays of the same order, even when exploration time increases.
 
\paragraph{\underline{Scalability with increasing query complexity}}
We now assess to which extent the approach scales with query complexity.
For that purpose, we consider a notion of recursive query complexity as the number of joined transitive closure relations. This slightly extends the usual notion of query complexity (traditionally measured as the number of joined relations) found in the litterature to also encompass recursion.  This accounts for the fact that both joins and recursions are complex operations, and also for the fact that they both significantly contribute to plan space size increase due to combinatorial explosions of equivalent plans (join ordering combined with rearranged recursions produce many more equivalent plans). 
Specifically, we consider a variable-length query which is a concatenation of transitive closures of relations:
$\mathcal{Q}_{r_i} = a_1^+/a_2^+/.../a_i^+$ for a given length $i$.
An increment in the query length increases both the number of recursions and joins, thus the query complexity, and therefore the size of the complete plan space. 
Specifically, $\mathcal{Q}_{r_i}$ contains $i$ recursions and $2i-1$ joins ($i-1$ joins at top level plus one join within each recursive part).

Figure~\ref{fig:scalability} shows a comparative analysis of plan enumeration speed of \mudag{} and \muenum{} when the query complexity increases. It shows the number of plans produced per millisecond (on the $y$ axis, in log scale) for a given query length (on the $x$ axis). %
Curves represent plan exploration speed. We observe that \mudag{}'s plan enumeration is always faster. Even more importantly, for the \mudag{} we observe an acceleration in enumeration, i.e. an increase of enumeration speed with respect to query complexity. The more complex is the query, the bigger are the complete plan space sizes, and the more \mudag{} becomes effective. This is because \mudag{}  transforms sets of terms of increasing sizes, which explains the progressive acceleration. \muenum{} explores the plan space term by term and hits a maximal exploration speed. In comparison, \mudag{}'s sharing and set-based exploration are clearly more effective. For example, for $\mathcal{Q}_{r_8}$ that contains 8 recursions and 7 top-level joins: \muenum{} explored 242K plans with a timeout of 10s, whereas \mudag{} explored 45M plans (2 orders of magnitude more).

\begin{figure}[H]
  \centering
  \vspace{-0.15cm}
    \includegraphics[width=0.93\linewidth]{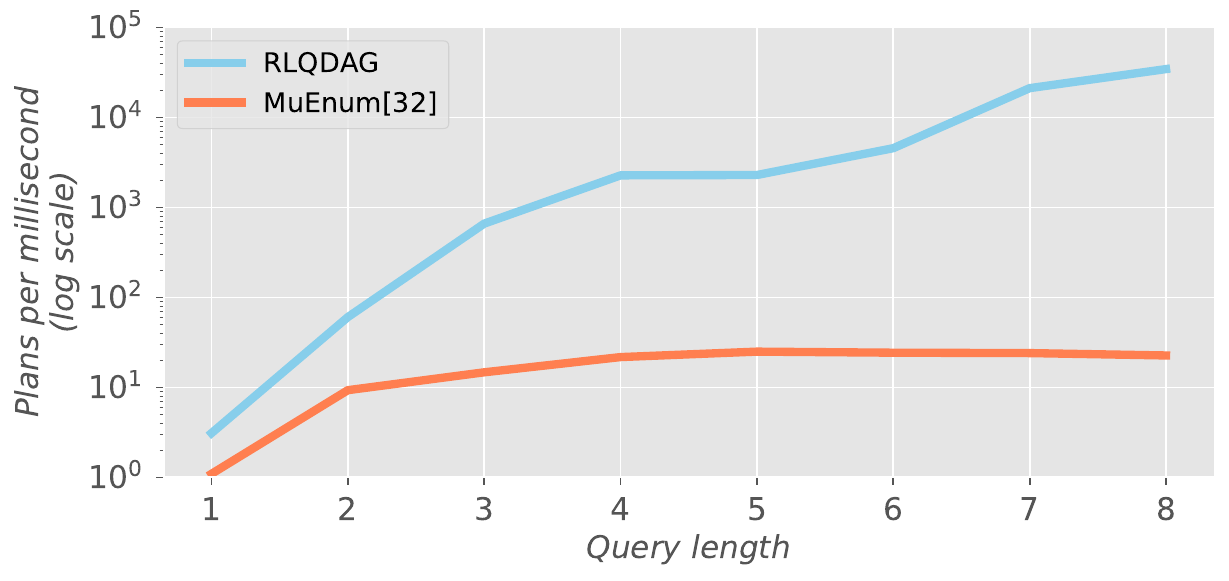}
    \vspace{-0.5cm}
    \caption{Enumeration with increasing query complexity.}\label{fig:scalability}
    \vspace{-0.3cm}
\end{figure} 

We now assess how relevant are faster explorations of larger plan spaces in terms of query evaluation performance.

\subsection{Relevance of large plan space exploration}

To answer \textbf{RQ2}, we use the same backend (PostgreSQL) in order to execute query plans generated by \muenum{} and \mudag{}.
For a given query, we set a time budget for plan space exploration, and we let \muenum{} and \mudag{} generate plan spaces for this same time budget.
Now, we pick the best estimated plan from each generated plan space. We use the same cost estimation \cite{lawal}, thus making relevant a head to head comparison. We measure and compare the times spent by PostgreSQL for evaluating the plans. %
\begin{figure}[h]
  \centering
  \includegraphics[width=0.9\linewidth]{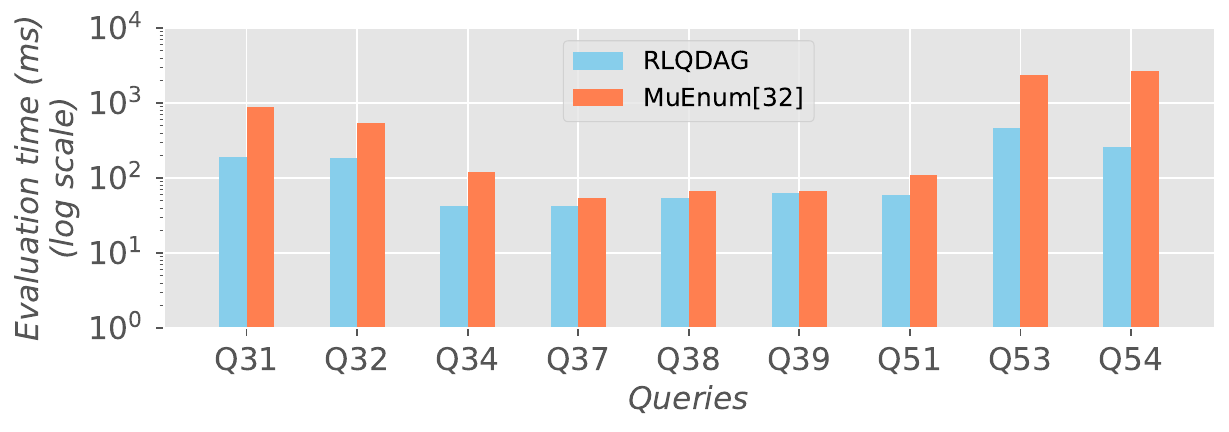} 
  \vspace{-0.4cm}
  \caption{Evaluating best estimated plans of spaces of Fig~\ref{fig:stacked-plans}.}
  \label{fig:eval-time} 
 \end{figure}  

For example, Fig.~\ref{fig:eval-time} illustrates the time spent in evaluating the best estimated plan taken from each of the explored plan spaces reported in Fig.~\ref{fig:stacked-plans}. Results show the benefits of exploring a much larger plan space: the \mudag{} approach always provides similar or better performance, which is a direct consequence of the availability of more efficient recursive plans in the larger plan space explored.

 \paragraph{\underline{Query running time with increasing exploration time budget}}
For a given query, we now inspect the impact on performance 
of increasing time budgets for plan space exploration. For that purpose, we trigger the plan space explorations with \mudag{} and \muenum{}  for different time budgets  ranging from 100ms to 500ms. 
We then measure the time spent in evaluating the best estimated plans taken from the corresponding plan spaces. We compare their respective performances. %
Fig.~\ref{subfig:q31-q53-evaluated-plans} shows the results obtained for queries Q31 and Q53. The $y$ axis shows the time spent (in log scale) in evaluating the best estimated plan obtained within the time budget shown on the $x$ axis. The sizes of corresponding plan spaces are given in Fig.~\ref{subfig:q31-q53-different-times}. 

For Q31, we observe that for both systems, the greater the exploration budget, the more efficient is the evaluation. We also observe that \mudag{} makes it possible to obtain more efficient plans much sooner (i.e. already with much smaller exploration time budgets). %
For Q53, \mudag{} is even more decisive since it explores the complete plan space in 100ms. The best possible estimated plan is thus already obtained, not needing any additional exploration budget.  In comparison, increasing time budgets for \muenum{} allows it to explore only small fractions of the whole plan space (as shown in Fig.~\ref{fig:stacked-plans}), failing to obtain a plan with better performance.

Overall, best estimated terms selected from larger plan spaces are more efficient. For any given exploration time budget, \mudag{} always produces more efficient terms than \muenum{}. 
This shows that in practice, larger plan spaces are very prone to contain more efficient recursive plans. This confirms the importance of efficiently exploring large recursive plan spaces.

\begin{figure}[h]
  \centering
  \vspace{-0.26cm}
  \includegraphics[width=0.9\linewidth]{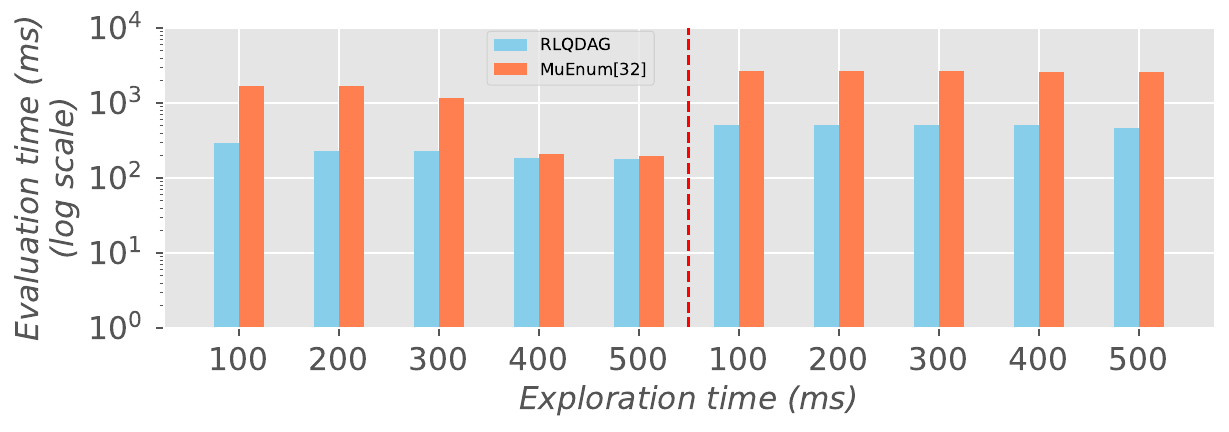}
  \vspace{-0.3cm}
  \caption{Evaluating best estimated plans of spaces of Fig~\ref{subfig:q31-q53-different-times}.}\label{subfig:q31-q53-evaluated-plans}
  \vspace{-0.26cm}
\end{figure}

Finally, we assess to which extent the availability of more relational algebraic plans can be useful in practice, when compared to other RDBMS and state-of-the-art approaches in graph query evaluation. For this purpose, we consider 2 engines based on relational algebra (\muenum{}~\cite{geneves-sigmod20} and Virtuoso~\cite{virtuoso} version 7.2.6.1), 3 native graph database engines (MilleniumDB~\cite{millenniumdb-vrgoc2021}, Neo4j~\cite{neo4j-program} version 4.4.11, Blazegraph~\cite{bigdata-blazegraph} version 2.1.6), and 3 plain-vanilla RDBMS capable of evaluating recursive SQL queries (PostgreSQL \cite{postgresql} version 15.1, MySQL \cite{mysql} version 8.1.0, SQLite \cite{sqlite} version 3.36.0), and one state-of-the-art Datalog engine (Souffl\'e \cite{souffle-cav16} version 2.4.1). We measure the time spent in query evaluation with each system.  
For each system, this includes the time spent in query optimization and the time spent in retrieving the whole set of query results. We set a timeout of 600s (10 min). Fig.~\ref{fig:comparisons-systems} shows the corresponding times spent for systems which were able to answer. If a system does not answer before the timeout, we consider that it is unable to answer, and it is absent from Fig.~\ref{fig:comparisons-systems} for the given query.

\begin{figure}[h]
    \centering
      \includegraphics[width=8.85cm, keepaspectratio,trim={0 0.8cm 0 0.3cm},clip]{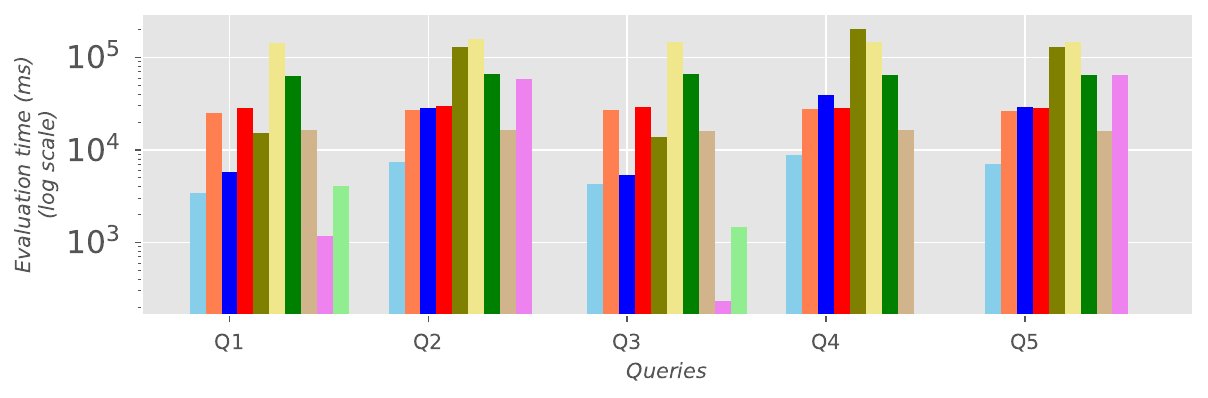}
      \includegraphics[width=8.85cm, keepaspectratio]{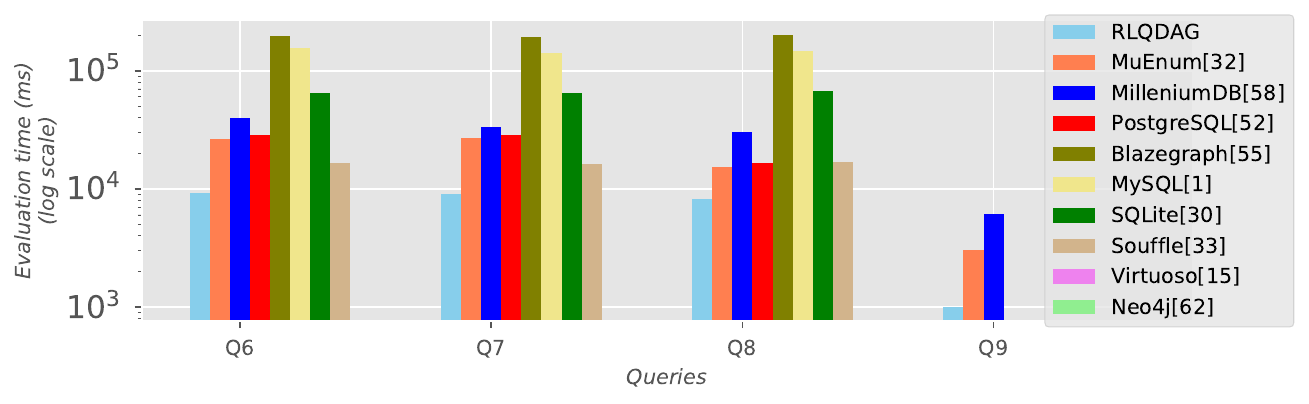}
      \vspace{-0.8cm}
      \caption{Comparative evaluation of third-party queries. 
      }\label{fig:comparisons-systems}
      \vspace{-0.4cm}
    \end{figure}

Results suggest that the availability of more plans -- thanks to the \mudag{} -- is beneficial to the relational-based approach.
More specifically, existing RDBMS consider recursive queries in a very restricted way. Most of them optimize recursion-free subexpressions, without being able to optimize the recursive part as a whole: they cannot make significant structural changes to the entire recursive query. This has an important consequence: they do explore much smaller plan spaces when compared to \mudag{}, potentially missing very efficient plans. In some cases, it even makes the relational-based approach more efficient than specialized graph engines. Likewise, Datalog engines like Souffl\'e are unable to generate many of the evaluation plans generated by the \mudag{}. For example they cannot merge recursions (as done in Section~\ref{sec:merging}). 

The plan spaces theoretically producible by \mudag{} and \muenum{} are the same. However, \muenum{} explores them term by term whereas \mudag{} explores them in a grouped manner which is much faster. This is because the compact representation of \mudag{} allows to transform sets of terms at once. As a consequence, \mudag{} makes it possible to explore in practice much larger portions of the theoretical plan space for the same time budget. 
Since transformations presented in Section~\ref{sec:transfo} are fully compositional, plans in which recursions are merged open further optimization opportunities. They unlock the exploration of many more plans in which for example joins and filters are pushed within merged recursions, etc. Such plans often happen to be much more efficient in practice.

\section{Conclusion and perspectives}
\label{section:conclusion}

We propose the \mudag{} for capturing and efficiently transforming sets of recursive relational terms. This is done by introducing annotated equivalence nodes, and a formal syntax and semantics for \mudag{} terms that enable the development of \mudag{} rewrite rules on a solid ground. \mudag{} rewrite rules transform sets of recursive terms while precisely describing how new subterms are created, attached, shared, and how new structural annotations are obtained with incremental updates. 
The proposed formalisation of the \mudag{} in terms of syntax and semantics provided a convenient -- if not instrumental -- means to develop transformations. It helps in defining expansions, and for detecting and fixing intricate transformational issues. We believe that this formalization can also serve in further extensions (such as groupBy and aggregations in the presence of recursion), thus contributing to the extensibility of the top-down transformational approach. 
Practical experiments with the \mudag{}  show the interest of exploring large plan spaces, and suggest that it represents an interesting foundation for efficiently enumerating recursive relational query plans. %

\newpage
\bibliographystyle{ACM-Reference-Format}
\bibliography{sample}

\newcommand{\proofpara}[1]{\vspace{0.5cm}\noindent\textbf{\underline{#1}}\vspace{0.3cm}\\}

\appendix
\section{Proof of Proposition 1} \label{appendix:proof1}

\setcounter{proprep}{0}
\begin{proprep}[Correctness]
    Let $\eqnode{\alpha}$ be a consistent \mudag{}, and $\alpha' = \expand{\eqnode{\alpha}}$, then we have $\sgamma{{\eqnode{\alpha}}}{\emptyset} \subseteq \sgamma{\alpha'}{\emptyset}$ and  $\alpha'$ is consistent. 
    \end{proprep}

    The proof of correctness is done by induction on the structure of $\alpha'$. For this purpose: (i) we focus on each newly created recursive term added by the expansion, (ii) we use the consistency hypothesis on the recursive term before expansion, and (iii) we show that the rearrangements made by the rewrite rules on this term preserve the correctness properties.

\emergencystretch 4em
\begin{proof}[Proof]
We break down the correctness proposition to three properties:
\begin{itemize} 
    \item[\textbf{(P1)}:] $\sgamma{{\eqnode{\alpha}}}{\emptyset} \subseteq \sgamma{\expand{\eqnode{\alpha}}}{\emptyset}$
    \item[\textbf{(P2)}:] $\expand{\eqnode{\alpha}}$ is well formed if $\eqnode{\alpha}$ is well formed
    \item[\textbf{(P3)}:] $\consppty{\expand{\eqnode{\alpha}}}$ if $\consppty{\eqnode{\alpha}}$
    where $\consppty{\eqnode{\alpha}}$ denotes "$\eqnode{\alpha}$ satisfies the property (2) of definition~\ref{def:consistency}"
\end{itemize}
In the following, we prove each of these properties.

\proofpara{Proof of property (P1).}
    We show this property by structural induction. Let $\eqnode{\alpha}$ be an RLQDAG, we suppose that every subterm $\eqnode{\alpha_\text{sub}}$ of $\eqnode{\alpha}$ verifies the following $\sgamma{\eqnode{\alpha_\text{sub}}}{E} \subseteq \sgamma{\expand{\eqnode{\alpha_\text{sub}}}}{E}$.
    It is sufficient to show that $\sgamma{{\eqnode{\alpha}}}{E} \subseteq \sgamma{\expand{\eqnode{\alpha}}}{E}$ for $\eqnode{\alpha} = \eqnode{d}$ because: 
    \begin{itemize}
        \item $\sgamma{\eqnode{\text{let Y } = \alpha_1 \text{ in } \alpha_2}}{E} = \sgamma{\eqnode{\alpha_2}}{E \cup \{ Y \mapsto  \alpha_1\}}$
        \item  $\sgamma{\eqnode{d, \alpha}}{E} = \sgamma{\eqnode{d}}{E} \cup \sgamma{\eqnode{\alpha_2}}{E}$
        \item $\expand{\eqnode{d, \alpha}} = \expand{\eqnode{d}}\cup \expand{\eqnode{\alpha}}$ (definition of expand~\ref{sec:expandfunction})
        \item So recursively, if the property is satisfied on $\eqnode{d}$ it is satisfied on $\eqnode{d1,d2,...}$
    \end{itemize}
    When $d$ is a leaf: $d = X$ (recursive variable) or $d = R$ (non recursive variable), we have $\expand{\eqnode{d}} = \eqnode{d}$, so the property holds.
    Otherwise, $d$ is the result of a unary operator: $d = \unaryop{\alphasub}$) or a binary operator: $d = \binaryop{\alphasub}{\alphasub'}$. In the following, we show the property for $d= \unaryop{\alphasub}$ (the proof for the the second case is identical).
    We can verify that for each rule $p$ of $\pf$, $\pp$, $\pa$, $\pj$, $\mf$, we have
    $p(\eqnode{\unaryop{\alphasub}}) = \eqnode{\unaryop{\expand{\eqnode{\alphasub}}}, ...}$. So by definition of $\sgamma{}{}$ and of $\expand{}$ we have,
    \begin{align*}
         \sgamma{\eqnode{\unaryop{\expand{\eqnode{\alphasub}}}}}{E} &\subseteq \sgamma{p(\eqnode{\unaryop{\alphasub}})}{E}  \\&\subseteq \sgamma{\expand{\eqnode{\unaryop{\alphasub}}}}{E} 
    \end{align*}
    So, $ \sgamma{\eqnode{\unaryop{\expand{\eqnode{\alphasub}}}}}{E} \subseteq \sgamma{\expand{\eqnode{d}}}{E} $ (1) \\
    We also have by definition of $\sgamma{}{}$
     $$\sgamma{\eqnode{\unaryop{\alphasub}}}{E} = \{\unaryop{t} | t \in \sgamma{\alphasub}{E}\}$$
    and $$\sgamma{\eqnode{\unaryop{\expand{\eqnode{\alphasub}}}}}{E} = \{\unaryop{t} | t \in \sgamma{\expand{\eqnode{\alphasub}}}{E}\}$$ 
    So, since $\sgamma{\eqnode{\alphasub}}{E} \subseteq \sgamma{\expand{\eqnode{\alphasub}}}{E}$ (induction hypothesis), we have $$\sgamma{\eqnode{d}}{E} = \sgamma{\eqnode{\unaryop{\alphasub}}}{E} \subseteq \sgamma{\eqnode{\unaryop{\expand{\eqnode{\alphasub}}}}}{E}$$ So $\sgamma{\eqnode{d}}{E} \subseteq \sgamma{\expand{\eqnode{d}}}{E}$ because (1).

    \proofpara{Proof of property (P2).}
    We define the property $\wellformed$ by the following:
    $$\alpha \text{ is } \wellformed~~ \text{ if }~~ \forall t, t' \in \salpha{\alpha}{E}, \quad \evaluation{t} =  \evaluation{t'}$$
    Let $\eqnode{\alpha}$ be an RLQDAG that is well formed. Here as well, it is sufficient to show that $\expand{\eqnode{\alpha}}$ is $\wellformed$ for $\eqnode{\alpha} = \eqnode{d}$ because for $\eqnode{\alpha} = \eqnode{d_1, d_2, ...}$:
    \begin{align*}
    \sbullet[1]~& \sgamma{\expand{\alpha}}{E} = \sgamma{\expand{\eqnode{d_1}}}{E} \cup \sgamma{\expand{\eqnode{d_2}}}{E} \cup ...\\
    \sbullet[1]~& \forall i, \forall t_i \in \sgamma{\eqnode{d_i}}{E}~\evaluation{t_i} = e \text{ for some fixed } e \\
     &\text{ because } t_i \in \sgamma{\eqnode{d}}{E} \text { and } \eqnode{d} \text{ is } \wellformed \\
    \sbullet[1]~& t_i \in \sgamma{\expand{\eqnode{d_i}}}{E} (\text{ because } (P1)) \\
    \sbullet[1]~& \text{So if }\sgamma{\expand{\eqnode{d_i}}}{E} \text { is } \wellformed, \forall s_i \in \sgamma{\expand{\eqnode{d_i}}}{E}, \\
    &\quad \evaluation{s_i}=e
    \end{align*}
    Let $d$ an operation node. In the following, we show that $\expand{\eqnode{d}}$ is $\wellformed$ when $\eqnode{d}$ is $\wellformed$.\\
    First, we can easily check that $\wellformed$ is satisfied on leaves ($d = X \text{ or } R$) because $\expand{\eqnode{d}} = \eqnode{d}$ and $\eqnode{d}$ is consistent.\\
    Second, we show by induction that $\expand{\eqnode{d}}$ is $\wellformed$, where $\eqnode{d}$ is $\wellformed$ and $d=\unaryop{\eqnode{\alpha}}$ or $d=\binaryop{\eqnode{\alpha}}{\eqnode{\alpha'}}$ for every operation node $\unaryop{}$ and $\binaryop{}{}$ and for any $\alpha$, $\alpha'$.\\
    We start by checking the base cases. When $\alpha$ is a leaf, we can easily check that $\expand{\eqnode{d}}$ is $\wellformed$ because we have $\expand{\eqnode{d}} = \eqnode{d}$ for all rewrite rules applied by $\expand{}$, except for the Codd join and union commutativity rules (ex. $\expand{\eqnode{A \NJoin B}} = \eqnode{A \NJoin B, B \NJoin A}$). In those cases as well, $\expand{\eqnode{d}}$ is $\wellformed$ because the new added term shares the same $\evaluation{}$ with the initial term. \\
    Then, we consider $\eqnode{\alpha}$, $\eqnode{\alpha'}$ two RLQDAGs and we suppose for that any subterm $\alpha_\text{sub}$ of $\alpha$ and any subterm $\alpha_\text{sub}'$ of $\alpha'$, we have the following: $\expand{\eqnode{\unaryop{\alpha_\text{sub}}}}$ and $\expand{\eqnode{\binaryop{\alpha_\text{sub}}{\alpha_\text{sub}'}}}$ are well formed when the terms inside $\expand{}$ are well formed, for every $\unaryop{}$ and for every $\binaryop{}{}$. We have:
    \begin{align*}
     \expand{\eqnode{d}} = &\pf{\eqnode{d}} ~\cup~ \pp{\eqnode{d}} ~\cup~ \pa{\eqnode{d}}\\
     & \cup~ \pj{\eqnode{d}} ~\cup~ \mf{\eqnode{d}} ~\cup~ \allCodd{\eqnode{d}}  
    \end{align*}
    We suppose that $\allCodd$ rules are correct, meaning that all terms $t_c \in \sgamma{\allCodd{\eqnode{d}}}{E}$ share the same $\evaluation{}$. They also share the same $\evaluation{}$ with the initial terms $t_i$ in $\sgamma{\eqnode{d}}{E}$. We have shown in the proof of $(P1)$ that for $t_i \in \sgamma{\eqnode{d}}{E}$, $t_i \in \sgamma{p{\eqnode{d}}}{E}$ for every fixpoint rule $p$ ($\pf{}$, $\pp{}$, $\pa{}$, $\pj{}$, $\mf{}$). This means that for $\expand{\eqnode{d}}$ to be $\wellformed$, it is sufficient to show that $p{\eqnode{d}}$ is $\wellformed$ for every fixpoint rule $p$. In the following, we show the proof for the rules $\pf$ and $\pj$. The proof for the rest of the rules is similar.\\
    We first start by showing the following property denoted by $\wfexpandp$:
    \begin{small}
    \begin{align*}
    &\forall \unaryop{}, \binaryop{}{}, \forall \beta, \beta' \text{where } \beta \text{ (resp. }\beta')\\
    & \qquad  \text{ is either }\alpha \text{ or any subterm of } \alpha \text{ (resp.} \alpha' \text{ or any subterm of }\alpha')) \\
    &\quad\unaryop{\eqnode{\beta}} \text{ is } \wellformed \implies \unaryop{\expand{\eqnode{\beta}}} \text{ is } \wellformed\\
    &\quad\binaryop{\eqnode{\beta}}{\eqnode{\beta'}} \text{ is } \wellformed \implies \binaryop{\expand{\eqnode{\beta}}}{\expand{\eqnode{\beta'}}} \text{ is } \wellformed
    \end{align*}
    \end{small}
    Let us suppose $\unaryop{\eqnode{\beta}}$ is $\wellformed$. We put $\eqnode{\beta} = \eqnode{d_1, d_2, ...}$.\\
    We have 
    \begin{align*}
    &\salpha{\unaryop{\expand{\beta}}}{E} =\\
    &\quad\salpha{\unaryop{\expand{\eqnode{d_1}}}}{E} \cup  \salpha{\unaryop{\expand{\eqnode{d_2}}}}{E} \cup ...
    \end{align*}
    Let $t_i \in \salpha{\unaryop{\eqnode{d_i}}}{E}$, we have $t_i \in \salpha{\unaryop{\eqnode{\beta}}}{E}$, and since $\eqnode{\unaryop{\eqnode{\beta}}}$ is $\wellformed$, then the terms $t_i$ share the same evaluation: $\evaluation{t_i} = e$. We also have, for each $i$, terms in $\salpha{\expand{\eqnode{d_i}}}{E}$ share the same $\evaluation{}$ $e_i$ (by the induction hypothesis on $d_i$ that is either a leaf or $\unaryop{\beta_{\text{sub}}}$ or $\binaryop{\beta_{\text{sub}}}{\beta_{\text{sub}}'}$). So terms in $\salpha{\unaryop{\expand{d_i}}}{E}$ share the same $\evaluation{}$ (application of $\unaryop{}$ on the common $e_i$). 
    Since $t_i$ also belongs to $\salpha{\unaryop{\expand{\eqnode{d_i}}}}{E}$ (because $(P1)$ and definition of $\salpha{}{}$), then all terms in $\salpha{\unaryop{\expand{\eqnode{d_i}}}}{E}$ share the same evaluation $e$ across all $i$. Hence $\salpha{\unaryop{\expand{\eqnode{\beta}}}}{E}$ is $\wellformed$.\\
    The proof of $\wfexpandp$ for the case of $\binaryop{}{}$ is identical.

    Now we show that $\pf{\eqnode{d}}$ and $\pj{\eqnode{d}}$ are $\wellformed$ when $\eqnode{d}$ is $\wellformed$.\\
    If $d$ is not a filter on a fixpoint (resp. if $d$ is not a join with a fixpoint), we have $\pf{\eqnode{d}}$ (resp. $\pj{\eqnode{d}}$) corresponds to the following:
    \begin{small}
    $
    \left\{\begin{array}{ll}
        & \unaryop{\expand{\eqnode{\alpha}}} \text{ when } d = \unaryop{\eqnode{\alpha}}\\ 
        &\binaryop{\expand{\eqnode{\alpha}}}{\expand{\eqnode{\alpha'}}} ~~ \text{ when } d = \binaryop{\eqnode{\alpha}}{\eqnode{\alpha'}}  
    \end{array}\right.$\\
    \end{small}
    So, in those cases, $\pf{\eqnode{d}}$ and $\pj{\eqnode{d}}$ are $\wellformed$ because $\wfexpandp$.\\
    If $d = \filt{\eqnode{\alpha}}$ where $\eqnode{\alpha}=\eqnode{\fixptdag{\gamma}{\alpha_\text{sub}}, \alpha_\text{sub}'}$ (we consider the more generic form where the fixpoint is not alone in an alternative), we have:
    \begin{small}
        $\pf{\eqnode{\filtaf{f}{\eqnode{\mu X.~ \gamma~ \cup~ \eqnode{\alpha_\text{sub}}_{\mathfrak{D}}^{\mathfrak{R}},~\alpha_\text{sub}'}}}}  =\newline
        \left\{\begin{array}{ll}
            \sbullet[1.1] &\hm{\bigr[} ~\pushedfb{}, ~~ \filt{\expand{\eqnode{\alpha}}}, ~~ \expand{\filt{\eqnode{\alpha_\text{sub}'}}} ~\hm{\bigr]}  \\
            &\quad \text{ when }~~  \filtColumns{f} \cap \mathfrak{D} = \emptyset \\ 
         \sbullet[1.1] &\hm{\bigr[} ~\filt{\expand{\eqnode{\alpha}}} ~, ~\ \expand{\filt{\eqnode{\alpha_\text{sub}'}}} ~\hm{\bigr]} ~~ \text{ otherwise}  
        \end{array}\right.$\\
        where:
        $\pushedfb{} = {\mu X'.~ \text{expand}(\eqnode{\filtaf{f}{\gamma}}) ~ \cup~ \eqnode{\expand{\subst{\alpha_\text{sub}}{X}{X'}}}_{\mathfrak{D}}^{\mathfrak{R}} }$\\
    \end{small}

Let $t \in \salpha{\alpha}{E}$ and $t' \in \salpha{\alpha_\text{sub}'}{E}$. We have $\filt{t} \in  \salpha{\filt{\eqnode{\alpha}}}{E}$ and $\filt{t'} \in  \salpha{\filt{\eqnode{\alpha_\text{sub}'}}}{E}$. We also have $\evaluation{\filt{t}} = \evaluation{\filt{t'}}$ because $\filt{t},\filt{t'} \in \salpha{d}{E}$ and $\eqnode{d}$ is $\wellformed$. Since $\eqnode{\filt{\expand{\eqnode{\alpha}}}}$ is $\wellformed$ because $\wfexpandp$ and $\expand{\filt{\eqnode{\alpha_\text{sub}'}}}$ is $\wellformed$ (induction hypothesis), then all terms in each of $\salpha{\filt{\expand{\eqnode{\alpha}}}}{E}$ and $\salpha{\expand{\filt{\eqnode{\alpha_\text{sub}'}}}}{E}$ evaluate to the same thing, and since, by $(P1)$, the first contains $\filt{t}$ and the second contains $\filt{t'}$ which share the same $\evaluation{}$, then all terms in $\salpha{\pf{\eqnode{d}}}{E}$ share the same $\evaluation{}$ (meaning $\pf{\eqnode{d}}$ is $\wellformed$) when $\filtColumns{f} \cap \mathfrak{D} \neq \emptyset$.\\
We have $\forall t'' \in \salpha{\pushedfb{}}{E},~~~ \exists t_\gamma \in \salpha{\gamma}{E}, t_\text{sub} \in \salpha{\alpha_\text{sub}}{E}$ such that $t'' = \fixpt{\filt{t_\gamma} \cup t_\text{sub}}$. We also have $$\evaluation{t''} = \evaluation{\filt{\fixpt{t_\gamma \cup t_\text{sub}}}} \text{ when } \filtColumns{f} \cap \mathfrak{D} = \emptyset~[32]. $$
Since $$\filt{\fixpt{t_\gamma \cup t_\text{sub}}} \in \salpha{\filt{\fixptdag{\gamma}{\alpha_\text{sub}}}}{E} \subset \salpha{d}{E}$$ then all terms in $\pushedfb{}$ share the same $\evaluation{}$ with the rest of the terms in $\pf{\eqnode{d}}$. Hence $\pf{\eqnode{d}}$ is $\wellformed$. 

If $d = \beta \NJoin \eqnode{\alpha}$ where $\eqnode{\alpha} = \eqnode{\fixptdag{\gamma}{\alpha_\text{sub}}, \alpha_\text{sub}'}$, we have\\
\begin{small}
    $\pj{\eqnode{\beta \NJoin \eqnode{\alpha} }} =\newline \mbox{ }~\hspace{0.3cm}\left\{\begin{array}{ll}
        \sbullet[1.1] & \eqnode{\text{let const} = \gamma ~\text{in}~
            \pushedfb{}  ~, \expand{\beta} \NJoin \expand{\eqnode{\alpha}}, \\
          & \quad \expand{\beta} \NJoin \expand{\eqnode{\alpha_\text{sub}'}} ~} \qquad \text{when } \condpj \text{ is true}\\  
       \sbullet[1.1] & \eqnode{\expand{\beta} \NJoin \expand{\eqnode{\alpha}}, \\
       & \quad \expand{\beta} \NJoin \expand{\eqnode{\alpha_\text{sub}'}} ~} \qquad \text{otherwise}  
      \end{array}\right.$\\
    where:
    $\pushedfb{} = \fixptdag{\expand{\eqnode{\beta \NJoin \text{const}}}}{\expand{\subst{\alpha_\text{sub}}{X}{X'}}}[\mathfrak{D}][\mathfrak{R}] $\\
    $\condpj = \type{\beta} \cap \mathfrak{D} = \emptyset ~and~ \type{\beta} \backslash \type{\gamma} \thickspace  \cap  \mathfrak{R} = \emptyset$ 
\end{small}

Let $t_\alpha \in \salpha{\alpha}{E}$, $t_\beta \in \salpha{\beta}{E}$ and $t_\alpha' \in \salpha{\alpha_\text{sub}'}{E}$. We have $t_\beta \NJoin t_\alpha \in  \salpha{\beta \NJoin {\eqnode{\alpha}}}{E}$ and $t_\beta \NJoin t_\alpha' \in  \salpha{\beta \NJoin \eqnode{\alpha_\text{sub}'}}{E}$. We also have $\evaluation{t_\beta \NJoin t_\alpha} = \evaluation{t_\beta \NJoin t_\alpha'}$ because $(t_\beta \NJoin t_\alpha)\in \salpha{d}{E}$ and $(t_\beta \NJoin t_\alpha') \in \salpha{\eqnode{d}}{E}$ and $\eqnode{d}$ is $\wellformed$. Since $\eqnode{\expand{\beta} \NJoin \expand{\eqnode{\alpha}}}$ is $\wellformed$ because $\wfexpandp$ and $\eqnode{\expand{\beta} \NJoin \expand{\eqnode{\alpha_\text{sub}'}}}$ is $\wellformed$ because $\wfexpandp$, then all terms in each of $\salpha{\expand{\beta} \NJoin \expand{\eqnode{\alpha}}}{E}$ and $\salpha{\expand{\beta} \NJoin \expand{\eqnode{\alpha_\text{sub}'}}}{E}$ evaluate to the same thing, and since, by $(P1)$, the first contains $t_\beta \NJoin t_\alpha$ and the second $t_\beta \NJoin t_\alpha'$, and these two terms share the same $\evaluation{}$, then all terms in $\salpha{\pj{\eqnode{d}}}{E}$ share the same $\evaluation{}$ (meaning $\pj{\eqnode{d}}$ is $\wellformed$) when the condition $\condpj$ is false.\\
We have $\forall t \in \salpha{\pushedfb{}}{E},~~~ \exists t_\beta \in \salpha{\beta}{E},~ t_\gamma \in \salpha{\gamma}{E},~ t_\text{sub} \in \salpha{\alpha_\text{sub}}{E}$ such that $t = \fixpt{t_\beta \NJoin t_\gamma \cup t_\text{sub}}$. We also have $\evaluation{t} = \evaluation{t_\beta \NJoin \fixpt{t_\gamma \cup t_\text{sub}}}$ when $\condpj$ is true~\cite{geneves-sigmod20}. Since $$t_\beta \NJoin \fixpt{t_\gamma \cup t_\text{sub}} \in \salpha{\beta \NJoin \fixptdag{\gamma}{\alpha_\text{sub}}}{E} \subset \salpha{d}{E}$$ then all terms in $\pushedfb{}$ share the same $\evaluation{}$ with the rest of the terms in $\pj{\eqnode{d}}$. Hence $\pj{\eqnode{d}}$ is $\wellformed$.

\proofpara{Proof of property (P3).}
    We first show the following:
     $$ \pconsexpand: \forall \eqnode{\alpha} \text{ such that } \consistent{\alpha}{X}{\mathfrak{D}}{\mathfrak{R}}, \consistent{\expand{\eqnode{\alpha}}}{X}{\mathfrak{D}}{\mathfrak{R}}$$
    Let us consider $\eqnode{\alpha} = \eqnode{d_1, d_2, ...}$. We have: \\
    $\expand{\eqnode{\alpha}} = \eqnode{\expand{\eqnode{d_1}}} \cup \eqnode{\expand{\eqnode{d_2}}} \cup ...$\\
    It is sufficient to show that $\consistent{\expand{\eqnode{d}}}{X}{\mathfrak{D}}{\mathfrak{R}}$ for every $\eqnode{d}$ verifying $\consistent{d}{X}{\mathfrak{D}}{\mathfrak{R}}$ because, for each $i$, terms in $\expand{\eqnode{d_i}}$ would share the same $\destab$ and $\strict$, and since terms in $d_i$ are found in $\expand{\eqnode{d_i}}$ $(P1)$, and they share the same $\destab$ and $\strict$ accross all $i$ (because terms in $d_i$ are in $\alpha$ and $\consistent{\alpha}{X}{\mathfrak{D}}{\mathfrak{R}}$), then $\expand{\eqnode{d_i}}$ would share the same $\destab$ and $\strict$ across all $i$. In the following, we prove that for all $\eqnode{d}$ such that $\consistent{d}{X}{\mathfrak{D}}{\mathfrak{R}}$, $\consistent{\expand{\eqnode{d}}}{X}{\mathfrak{D}}{\mathfrak{R}}$. 
    First, for $d$ a leaf ($X$ or $R$) $\expand{\eqnode{d}} = \eqnode{d}$, hence the property is satisfied.
    Second, we prove by structural induction that $\forall \alpha, \alpha' ~~~ \forall \unaryop{}, \binaryop{}{},$ if $\consistent{\eqnode{d}}{X}{\mathfrak{D}}{\mathfrak{R}}$, then $\consistent{\expand{\eqnode{d}}}{X}{\mathfrak{D}}{\mathfrak{R}}$ for $d=\unaryop{\alpha}$ and $d=\binaryop{\alpha}{\alpha'}$. \\
    We show the proof for $d= \filt{\alpha}$ and the proof for the other operators is similar.\\
    The property is satisfied for $\alpha$ a leaf, because $\expand{\eqnode{d}} = \eqnode{d}$.\\ \emergencystretch 5em
    Let $\alpha$ such that $\consistent{\filt{\eqnode{\alpha}}}{X}{\mathfrak{D}}{\mathfrak{R}}$. We show that $\consistent{\expand{\filt{\eqnode{\alpha}}}}{X}{\mathfrak{D}}{\mathfrak{R}}$.\\
    We have $\expand{\filt{\eqnode{\alpha}}} = \pf{\eqnode{\alpha}} \cup \allCodd{\filt{\eqnode{\alpha}}}$. In the following, we show that $\consistent{\pf{\filt{\eqnode{\alpha}}}}{X}{\mathfrak{D}}{\mathfrak{R}}$ and we can similarly verify that every rule $p$ in $\allCodd{}$ also satisfies $\consistent{p(\filt{\eqnode{\alpha}})}{X}{\mathfrak{D}}{\mathfrak{R}}$.
    If $\alpha$ is not a fixpoint, $\pf{\filt{\eqnode{\alpha}}} = \filt{\expand{\eqnode{\alpha}}} = \filt{\expand{\eqnode{d_1}}} \cup \filt{\expand{\eqnode{d_2}}},...$ for $\eqnode{\alpha} = \eqnode{d_1, d_2, ...}$.
    We have, for each $i$, terms in $\expand{\eqnode{d_i}}$ share the same $\destab$ and $\strict$ (induction hypothesis with $d_i$ is either $\unaryop{\alpha_\text{sub}}$ or $\binaryop{\alpha_\text{sub}}{\alpha_\text{sub}'}$), which means that, for each $i$, terms in $\filt{\expand{\eqnode{d_i}}}$ share the same $\destab$ $\mathfrak{D_i}$ and $\strict$ $\mathfrak{R_i}$. Since terms in $\filt{d_i}$ are found in $\filt{\expand{\eqnode{d_i}}}$ $(P1)$, and they share the same $\destab$ $\mathfrak{D}$ and $\strict$ $\mathfrak{R}$ accross all $i$ (because terms in $\filt{d_i}$ are in $\filt{\alpha}$ and $\consistent{\filt{\eqnode{\alpha}}}{X}{\mathfrak{D}}{\mathfrak{R}}$), then $\filt{\expand{\eqnode{d_i}}}$ share the same $\destab$ $\mathfrak{D}$ and $\strict$ $\mathfrak{R}$ across all $i$. Hence $\consistent{\pf{\filt{\eqnode{\alpha}}}}{X}{\mathfrak{D}}{\mathfrak{R}}$. \\
    Otherwise, if $\alpha$ is a fixpoint we have \\
    \begin{small}
    $\pf{\eqnode{\filtaf{f}{\eqnode{\mu Z.~ \gamma~ \cup~ \eqnode{\alpha_\text{sub}}_{\mathfrak{D'}}^{\mathfrak{R'}},~\alpha_\text{sub}'}}}}  =\newline
    \left\{\begin{array}{ll}
        \sbullet[1.1] &\hm{\bigr[} ~\filt{\expand{\eqnode{\alpha}}}, ~~ \pushedfb{}, ~~  \expand{\filt{\eqnode{\alpha_\text{sub}'}}} ~\hm{\bigr]}  \\
        &\quad\text{ when }~~  \filtColumns{f} \cap \mathfrak{D'} = \emptyset \\ 
     \sbullet[1.1] &\hm{\bigr[} ~\filt{\expand{\eqnode{\alpha}}} ~, ~\ \expand{\filt{\eqnode{\alpha_\text{sub}'}}} ~\hm{\bigr]} ~~ \text{ otherwise}  
    \end{array}\right.$\\
    where:
    $\pushedfb{} = {\mu Z'.~ \text{expand}(\eqnode{\filtaf{f}{\gamma}}) ~ \cup~ \eqnode{\expand{\subst{\alpha_\text{sub}}{Z}{Z'}}}_{\mathfrak{D'}}^{\mathfrak{R'}} }$\\
    \end{small}
    Using the same reasoning used for the case of $\alpha$ not a fixpoint as well as the induction hypothesis on $\filt{\eqnode{\alpha_\text{sub}'}}$, we can show that $\consistent{\filt{\expand{\eqnode{\alpha}}}}{X}{\mathfrak{D}}{\mathfrak{R}}$ and $\consistent{\expand{\filt{\eqnode{\alpha_\text{sub}'}}}}{X}{\mathfrak{D}}{\mathfrak{R}}$. We next show $\consistent{\pushedfb{}}{X}{\mathfrak{D}}{\mathfrak{R}}$.\\
    Let $R_\gamma = \strict(\gamma, X)$ and $R_{\alpha_\text{sub}} = \strict(\alpha_\text{sub}, X)$. \\
    By the induction hypothesis, we have $$\strict(\expand{\filt{\gamma}}, X) = \strict(\filt{\gamma}, X) = \filtColumns{f} \cup R_\gamma$$ 
    Using a reasoning similar to the one used for the case of $\alpha$ not a fixpoint, we can show that $\strict(\expand{\eqnode{\alpha_\text{sub}}}, X) = R_{\alpha_\text{sub}}$.
    So $\strict(\pushedfb{}) = \filtColumns{f} \cup R_\gamma \cup R_{\alpha_\text{sub}}$. We have, by definition of $\strict$, $\strict(\filt{\fixptdag[Z]{\gamma}{\alpha_\text{sub}}[\mathfrak{D'}][\mathfrak{R'}]}, X) = \filtColumns{f} \cup R_\gamma \cup R_{\alpha_\text{sub}}$. So  $\strict(\pushedfb, X) = \strict(\filt{\fixptdag[Z]{\gamma}{\alpha_\text{sub}}[\mathfrak{D'}][\mathfrak{R'}]}, X)$. So $\strict(\pushedfb, X) = \mathfrak{R}$ because terms in $\filt{\fixptdag[Z]{\gamma}{\alpha_\text{sub}}[\mathfrak{D'}][\mathfrak{R'}]}$ are in $\filt{\eqnode{\alpha}}$ and $\consistent{\filt{\eqnode{\alpha}}}{X}{\mathfrak{D}}{\mathfrak{R}}$. Similarly, $$\mathfrak{D}=\destab(\filt{\fixptdag[Z]{\gamma}{\alpha_\text{sub}}[\mathfrak{D'}][\mathfrak{R'}]}, X) = \emptyset = \destab(\pushedfb{}, X)$$
    We have shown that each of $\pushedfb{}$, $\filt{\expand{\eqnode{\alpha_\text{sub}}}}$, and $\expand{\filt{\eqnode{\alpha_\text{sub}'}}}$ have $\mathfrak{R}$ as $\strict$ and $\mathfrak{D}$ as $\destab$. Hence, $\consistent{\pf{\eqnode{\filt{\eqnode{\alpha}}}}}{X}{\mathfrak{D}}{\mathfrak{R}}$.

   We now prove $(P3)$. For $\alpha = \eqnode{d_1, d_2, ...}$,  $\expand{\eqnode{\alpha}} = \eqnode{\expand{\eqnode{d_1}} \cup \expand{\eqnode{d_2}}, ...}$. So it is sufficient to show $\consppty{\expand{\eqnode{d}}}$ for any $\eqnode{d}$ satisfying $\consppty{\eqnode{d}}$ (because any fixpoint in $\expand{\eqnode{\alpha}}$ is found in one of $\expand{\eqnode{d_i}}$).
    We have $\expand{\eqnode{d}} = \applyAll{d}$. So it is sufficient to show, for every rule, that any term $\eqnode{\alpha}_{\mathfrak{D}}^{\mathfrak{R}}$ produced by this rule satisfies $\consistent{\alpha}{X}{\mathfrak{D}}{\mathfrak{R}}$. The only rules that produce new terms of the shape $\eqnode{\alpha}_{\mathfrak{D}}^{\mathfrak{R}}$ are the fixpoint rules ($\pf$, $\pa$, $\pp$, $\pj$, $\mf$). We can verify that all these rules only produce the following such terms:
    \begin{itemize}
    \item $\anneqnode{\alpha}$ that already exists in $\eqnode{d}$ which already satisfies the property.
    \item $\anneqnode{\expand{\eqnode{\alpha}}}[\mathfrak{D}][\mathfrak{R}]$ from a fixpoint $\fixptdag{\gamma}{\alpha}$. We have $\consistent{\alpha}{X}{\mathfrak{D}}{\mathfrak{R}}$, so $\consistent{\expand{\alpha}}{X}{\mathfrak{D}}{\mathfrak{R}}$ $\pconsexpand$.
    \item $\anneqnode{\expand{\eqnode{\alpha_1 \cup \alpha_2}}}[\mathfrak{D}_1 \cup \mathfrak{D}_2][\mathfrak{R}_1 \cup \mathfrak{R}_2]$ from from two fixpoints: $\fixptdag{\gamma}{\alpha}[\mathfrak{D}_1][\mathfrak{R}_1]$ and $\fixptdag{\gamma'}{\alpha'}[\mathfrak{D}_2][\mathfrak{R}_2]$. We have $\consistent{\alpha}{X}{\mathfrak{D}_1}{\mathfrak{R}_1}$ and $\consistent{\alpha'}{X}{\mathfrak{D}_2}{\mathfrak{R}_2}$, so $\consistent{\alpha \cup \alpha'}{X}{\mathfrak{D}_1 \cup \mathfrak{D}_2}{\mathfrak{R}_1 \cup \mathfrak{R}_2}$, so $\consistent{\expand{\alpha \cup \alpha'}}{X}{\mathfrak{D}_1 \cup \mathfrak{D}_2}{\mathfrak{R}_1 \cup \mathfrak{R}_2}$ $\pconsexpand$.
    \end{itemize}

\end{proof}

\end{document}